%% file: neurips_2023.tex
\documentclass{article}

\usepackage[final]{neurips_2023}

\usepackage[utf8]{inputenc} %
\usepackage[T1]{fontenc}    %
\usepackage{hyperref}       %
\usepackage{url}            %
\usepackage{booktabs}       %
\usepackage{amsfonts}       %
\usepackage{nicefrac}       %
\usepackage{microtype}      %
\usepackage{xcolor}         %
\usepackage{algorithm,algpseudocode} %

\input{neurips_definitions}

\newcommand{\CIC}{\mathsf{CIC}}

\newcommand{\I}{\mathsf{I}}
\newcommand{\Hsf}{\mathsf{H}}
\newcommand{\AND}{\mathrm{AND}}
\newcommand{\OPT}{\mathsf{OPT}}
\DeclareMathOperator{\TV}{\mathsf{TV}}
\usepackage{amsmath, mathtools}
\DeclarePairedDelimiter\parens{\lparen}{\rparen}
\DeclarePairedDelimiter\braces{\lbrace}{\rbrace}
\DeclarePairedDelimiter\bracks{\lbrack}{\rbrack}
\let\angle\Angle
\DeclarePairedDelimiter\angle{\langle}{\rangle}

\title{Sketching Algorithms for Sparse Dictionary Learning: PTAS and Turnstile Streaming}

\author{%
Gregory Dexter \\
Department of Computer Science \\
Purdue University \\
\texttt{gdexter@purdue.edu} \\
\And
Petros Drineas \\
Department of Computer Science \\
Purdue University \\
\texttt{pdrineas@purdue.edu} \\
\And
David P. Woodruff \\
Computer Science Department \\
Carnegie Mellon University \\
\texttt{dwoodruf@cs.cmu.edu} \\
\And
Taisuke Yasuda \\
Computer Science Department \\
Carnegie Mellon University \\
\texttt{taisukey@cs.cmu.edu} \\
}

\begin{document}

\maketitle

\begin{abstract}
Sketching algorithms have recently proven to be a powerful approach both for designing low-space streaming algorithms as well as fast polynomial time approximation schemes (PTAS). In this work, we develop new techniques to extend the applicability of sketching-based approaches to the \emph{sparse dictionary learning} and the \emph{Euclidean $k$-means clustering} problems. In particular, we initiate the study of the challenging setting where the dictionary/clustering \emph{assignment} for each of the $n$ input points must be output, which has surprisingly received little attention in prior work. On the fast algorithms front, we obtain a new approach for designing PTAS's for the $k$-means clustering problem, which generalizes to the first PTAS for the sparse dictionary learning problem. On the streaming algorithms front, we obtain new upper bounds and lower bounds for dictionary learning and $k$-means clustering. In particular, given a design matrix $\Ab\in\mathbb R^{n\times d}$ in a turnstile stream, we show an $\tilde O(nr/\epsilon^2 + dk/\epsilon)$ space upper bound for $r$-sparse dictionary learning of size $k$, an $\tilde O(n/\epsilon^2 + dk/\epsilon)$ space upper bound for $k$-means clustering, as well as an $\tilde O(n)$ space upper bound for $k$-means clustering on random order row insertion streams with a natural ``bounded sensitivity'' assumption. On the lower bounds side, we obtain a general $\tilde\Omega(n/\epsilon + dk/\epsilon)$ lower bound for $k$-means clustering, as well as an $\tilde\Omega(n/\epsilon^2)$ lower bound for algorithms which can estimate the cost of a single fixed set of candidate centers.
\end{abstract}

\section{Introduction}

A classic idea in machine learning and signal processing for efficiently handling large datasets is to approximate them by simpler or more structured surrogate datasets. Many methods in this direction have long been considered, including low rank approximation, which approximates a given dataset by one that lies on a low-dimensional subspace, $k$-means clustering, which approximates a given dataset by at most $k$ distinct points, and sparse dictionary learning \cite{OF1997}, which approximates a given dataset by linear combinations of elements of a small dictionary of size $k$ with $r$-sparse coefficient vectors (i.e., a vector with at most $r$ nonzero entries). We focus on the latter two problems in this work: %

\begin{definition}[$r$-sparse dictionary learning]
\label{def:sparse_dict}
Let $\{a^i\}_{i=1}^n\subseteq\mathbb R^d$ be a set of $n$ vectors in $d$ dimensions, and let $\Ab\in\mathbb R^{n\times d}$ be the matrix with the $i$th row set to $a^i$. Then for a matrix $\Xb\in\mathbb R^{n\times k}$ with $r$-sparse rows and a dictionary $\Db\in\mathbb R^{k\times d}$, we define the dictionary learning cost to be
\[
    \cost(\Xb,\Db) \coloneqq \norm{\Xb\Db - \Ab}_F^2
\]
In the \emph{$r$-sparse dictionary learning problem}, we seek to minimize $\cost(\Xb,\Db)$ over all $\Xb\in\mathcal X$ and $\Db\in\mathbb R^{k\times d}$, where $\mathcal X$ denotes the set of all $n\times k$ matrices with $r$-sparse rows.
\end{definition}

\begin{definition}[Euclidean $k$-means clustering]
\label{def:k-means}
Let $\{a^i\}_{i=1}^n\subseteq\mathbb R^d$ be a set of $n$ vectors in $d$ dimensions, and let $\Ab\in\mathbb R^{n\times d}$ be the matrix with the $i$th row set to $a^i$. Then, for a matrix $\Xb\in\mathbb R^{n\times k}$ with standard basis vectors in its rows and a set of centers $\Cb\in\mathbb R^{k\times d}$, we define the $k$-means clustering cost to be
\[
    \cost(\Xb,\Cb) \coloneqq \norm{\Xb\Cb-\Ab}_F^2.
\]
In the \emph{$k$-means clustering problem}, we seek to minimize $\cost(\Xb,\Cb)$ over all $\Xb\in\mathcal X$ and $\Cb\in\mathbb R^{k\times d}$, where $\mathcal X$ denotes the set of all $n\times k$ matrices with standard basis vectors as rows.
\end{definition}

While dictionary learning and clustering have found extraordinary success in various applications in practice, they are known to be computationally difficult problems to solve \citep{MNV2012, Nat1995}, and thus there has been intense focus on developing approximation algorithms and heuristics for these problems, such as those based on greedy methods \citep{Llo1982, DK2011} or convex relaxations \citep{DE2003, Fuc2004, CEMN2022}. 

In this work, we study algorithms for sparse dictionary learning and $k$-means clustering in two distinct settings via a unified set of techniques based on \emph{sketching}. Sketching \citep{woodruff2014sketching}, broadly speaking, refers to techniques for compressing large matrices by linear maps, and includes methods such as oblivious sketching and nonuniform sampling. Classically, sketching has been applied to design low-memory algorithms in the \emph{streaming setting}, when the input is presented to the algorithm as a sequence of updates. More recently, sketching has been shown to be invaluable for designing fast algorithms as well. In particular, there has been a line of work which shows how sketching techniques can be applied to obtain \emph{polynomial time approximation schemes} (PTAS) for a variety of NP-hard problems ranging from clustering \citep{feldman2007ptas} to weighted low rank approximation \citep{RSW2016} to tensor decompositions \citep{SWZ2019}. We study such sketching-based algorithms for sparse dictionary learning and Euclidean $k$-means clustering, both in the offline setting where we obtain the first PTAS for sparse dictionary learning, as well as in the turnstile streaming and other streaming models. In particular, in the streaming setting, we initiate the study of solving these problems in the setting where the algorithm must output the assignment of the points to the dictionary/clustering, which has received surprisingly little attention in prior work.

\subsection{Our contributions}

\subsubsection{PTAS for dictionary learning and clustering}

We start with a discussion of our results on designing fast PTAS's. Our main contribution that we highlight from this section is the \emph{first PTAS for sparse dictionary learning}, which also gives a new and simple approach towards designing a PTAS for $k$-means clustering.

A typical approach for designing PTAS's for shape fitting problems such as dictionary learning and clustering is to first find a smaller instance whose solution approximates the original instance, and then to solve the smaller instance using any algorithm, where even an inefficient algorithm will be tractable due to the smaller size of the instance. A representative work which takes such an approach for the $k$-means clustering problem is that of \citet{feldman2007ptas}, which uses \emph{coresets} to implement the first step of finding a smaller instance. Here, coresets for $k$-means clustering are a weighted subset of the original data points such that the cost of any candidate set of centers approximates the cost when applied to the original dataset. Furthermore, the size of this coreset can be taken to be $\poly(k/\epsilon)$, and thus solving for an optimal set of centers on this subset of points can be done in time independent of the number of points $n$. Due to this natural approach, there has been a long line of work on obtaining smaller coresets for $k$-means clustering \citep{FL2011, BFL2016, bachem2018one, CSS2021, CLSS2022, CLSSS2022}. 

On the other hand, for the sparse dictionary learning problem, similar results are strikingly lacking. The only previous work we are aware of is a coreset construction for the sparse dictionary learning problem due to \citet{feldman2013learning}. However, the construction of the coreset in this work requires an algorithm for computing an approximately optimal dictionary, which prevents its use in designing fast PTAS's to solve the dictionary learning problem in the first place. To address this problem, we first show that a completely different coreset technique due to \cite{tukan2022new} for the projective clustering problem can in fact be applied to the sparse dictionary learning problem. Notably, this technique uses John ellipsoids to construct coresets rather than using a nearly optimal solution to the dictionary learning problem, and thus avoids computing approximately optimal dictionaries. In turn, this allows us to obtain the first PTAS for the dictionary learning problem. Our argument additionally combines this coreset construction with a sparsity-counting technique together with polynomial system solvers \cite{renegar1992computational_decision, renegar1992computational_optimization} to efficiently solve a smaller version of the original problem. Our techniques also yield a new PTAS for $k$-means clustering, which is arguably simpler than prior approaches such as the algorithm of \cite{feldman2007ptas}. We give a full discussion of our results and techniques for our PTAS for sparse dictionary learning in Section \ref{sec:ptas}.

\subsubsection{Dictionary learning and clustering on streams}

As our next contribution, we study algorithms for dictionary learning and clustering in turnstile streams and other related models of streaming. In the turnstile streaming model, the input undergoes arbitrary entrywise insertions and deletions:

\begin{definition}[Turnstile stream]
\label{def:turnstile-stream}
We say that an input matrix $\Ab\in\mathbb R^{n\times d}$ is presented in a \emph{turnstile stream} if $\Ab$ is initialized to $0$ and receives entrywise updates $\Ab_{i,j}\gets \Ab_{i,j} + \Delta$ for $\Delta\in\mathbb R$.
\end{definition}

We initiate a systematic study of the dictionary learning and clustering problems in the setting where the assignment of the points to their sparse set of dictionary elements or clusters must be output together with the dictionary/cluster centers. Indeed, even for the popular Euclidean $k$-means clustering problem, almost all prior work that we are aware of only focus on outputting either only the cluster partitions, or the centers, but do not study the problem of recovering both. We address this problem by providing a dimensionality reduction technique that applies to $k$-means, sparse dictionary learning, and more generally to any problem of the form $\min_{\Xb\in\Xcal,\Db\in\R^{k\times d}}\|\Xb\Db - \Ab\|_F^2$.

A typical approach for designing low-space streaming algorithms for clustering is to apply the standard Johnson--Lindenstraus lemma \citep{JL1984, BZD2010, cohen2015dimensionality, becchetti2019oblivious, makarychev2019performance}. This result states that if $\Gb\in\mathbb R^{d\times s}$ is an appropriately scaled dense sub-Gaussian matrix for $s = O(\epsilon^{-2}\log(k/\epsilon))$, then for any partition of $\Ab$ into $k$ clusters, the $k$-means clustering cost of $\Ab\Gb$ approximates the $k$-means clustering cost of $\Ab$ up to a $(1\pm\epsilon)$ factor. Furthermore, $\Ab\Gb$ can be efficiently maintained in the turnstile streaming model (Definition \ref{def:turnstile-stream}) using just $ns = \tilde O(\epsilon^{-2}n)$ space, %
due to the linearity of the sketch $\Gb$. Note however that, na\"ively, we cannot retrieve the corresponding centers of a clustering found by this method, since we have only stored the $s$-dimensional sketches of the $n$ points, and additional information must be stored in order to retrieve $d$-dimensional cluster centers which achieve a $(1+\epsilon)$ approximation. In fact, we note in Theorem \ref{thm:cluster-center-lb} that there is in fact a $\tilde\Omega(dk/\epsilon)$ space lower bound if we wish to output centers $\Cb\in\mathbb R^{k\times d}$ which achieve a $(1+\epsilon)$ approximation, so the sketch $\Ab\Gb$ is \emph{provably} insufficient for outputting both a nearly optimal assignment $\Xb$ and centers $\Cb$ when $n = \tilde o(\epsilon dk)$. We give a full discussion of our approaches for sketching and streaming algorithms for $k$ means clustering and dictionary learning and how we overcome this problem in Sections \ref{sec:ptas} and \ref{sec:turnstile-streaming}.

On the other hand, a study of lower bounds for the $k$-means clustering problem in the streaming setting when the assignment of points must be output is notably lacking in prior work as well. The main challenge in this setting is in obtaining the right dependence on $n$ and $\epsilon$. Indeed, an $\Omega(n)$ lower bound is immediate, since the size of the output is at least $\Omega(n)$ when we need to output assignments of the $n$ points to its appropriate cluster (in fact, we show in Theorems \ref{thm:lb-cost} and \ref{thm:lb-center} that an $\Omega(n)$ lower bound follows even for outputting a constant factor approximation of the cost or centers). On the other hand, the previous upper bound using the Johnson--Lindenstrauss lemma to compute a nearly optimal assignment to clusters requires $\tilde O(\epsilon^{-2}n)$ bits of space. Note that there are many lower bounds that show that roughly $\epsilon^{-2}$ dimensions are required to apply the Johnson--Lindenstrauss lemma in various settings \cite{NN2014, KNW2010, LN2016, LN2017, makarychev2019performance}. However, it is not clear whether or not this implies that $\epsilon^{-2}$ bits must be stored for all $n$ points in order to cluster them to a $(1+\epsilon)$-approximately optimal clustering solution. Indeed, it may be possible that $\epsilon^{-2}$ bits are required only for much fewer than $n$ points, while the vast majority of the $n$ input points requires only $\tilde O(n)$ bits of space to assign to an approximately optimal center. 

We present two lower bounds to partially address the question of impossibility results for assigning points to clusters in turnstile streams. Our main lower bound result is the following, which establishes an $\tilde \Omega(\epsilon^{-1}n)$ lower bound to output a $(1+\epsilon)$-nearly optimal clustering. While this does not match the upper bound given by the Johnson--Lindenstrauss lemma, it shows that we cannot hope for a $\tilde O(n)$ upper bound in the turnstile streaming model in general.

\begin{restatable}[Informal restatement of Theorem \ref{thm:k-means-lb}]{theorem}{KMeansLB}
Let $k = d = \tilde O(1/\epsilon)$. Suppose a turnstile streaming algorithm outputs centers $\{\hat c^j\}_{j=1}^k\subseteq\mathbb R^d$ as well as assignments of $n$ points to the $k$ centers, which achieves a $(1+\epsilon)$-approximately optimal solution to the $k$-means clustering problem. Then, the algorithm must use at least $\tilde\Omega(n/\epsilon)$ bits of space over any constant number of passes.
\end{restatable}

As a second lower bound result, we also show that the Johnson-Lindenstrauss lemma is nearly tight if we require our algorithm to give a nearly optimal assignment of the input points to a fixed set of candidate centers. That is, we show in Theorem \ref{thm:strong-lb} that there is a fixed set of centers such that, if a turnstile streaming algorithm can assign each of the $n$ input points to a cluster such that the cost is at most $(1+\epsilon)$ times the cost of the optimal assignment, then at least $\Omega(\epsilon^{-2}n)$ bits must be stored. A more detailed discussion of our lower bounds is given in Section \ref{sec:lower-bounds}.

Finally, we show that under some natural settings, one can obtain upper bounds that circumvent the lower bounds presented above. Indeed, we show that if we work in the \emph{random order row arrival} streaming model, in which the input stream corresponds to the rows of $\Ab$ that arrive in a uniformly random order, then we can obtain upper bounds that depend on the \emph{maximum sensitivity} of the input stream, %
and in particular, we obtain an upper bound using only $\tilde O(n)$ bits of space if the maximum sensitivity is sufficiently small (Theorem \ref{thm:random-order-bounded-sensitivity}). Here, a bounded sensitivity assumption states that there are no points that can take up a significant fraction of the objective function, and can also be interpreted as a way to formalize a ``well-clustered'' instance.

\input{neurips_upper}

\input{neurips_lower}

\section{Open directions}

We conclude with several questions left open by our work. 
\begin{enumerate}
    \item In our PTAS for sparse dictionary learning of Theorem \ref{thm:sparse_dict_ptas}, can the bit complexity assumption be removed?
    \item In the turnstile streaming setting, our main question is settling the space complexity of $k$-means clustering with assignments. Currently, the upper bound is $\tilde O(n/\epsilon^2)$ bits whereas our lower bound in Theorem \ref{thm:k-means-lb} is $\tilde\Omega(n/\epsilon)$ bits. Can this $\epsilon$ factor gap be closed by improving the upper bound or the lower bound?
    \item In random order streaming model, we gave an $k$-means clustering upper bound using a bounded sensitivity assumption in Theorem \ref{thm:random-order-bounded-sensitivity}. Can this assumption be removed? What upper bounds and lower bound are possible in this model? 
\end{enumerate}

\begin{ack}
We thank the anonymous reviewers for useful feedback on improving the presentation of this work. Petros Drineas and Gregory Dexter were partially supported by NSF AF 1814041, NSF FRG 1760353, and DOE-SC0022085. David P.\ Woodruff and Taisuke Yasuda were supported by a Simons Investigator Award.

\end{ack}

\bibliography{neurips_bibliography}

\newpage

\appendix

\input{appendix}

\section{Information Theory Preliminaries}
\label{sec:info-theory}

\begin{definition}[Entropy and Mutual Information]
Let $X,Y, Z$ be discrete random variables. Then, the \emph{entropy} of $X$ is defined as
\[
    \Hsf(X) \coloneqq \sum_x \Pr[X=x] \log\frac1{\Pr[X=x]}
\]
and the \emph{conditional entropy of $X$ given $Y$} is defined as
\[
    \Hsf(X\mid Y) \coloneqq \EE_{y\sim Y}[ \Hsf(X\mid Y=y) ]
\]
The \emph{mutual information between $X$ and $Y$} is defined as
\[
    \I(X;Y) \coloneqq \Hsf(X) - \Hsf(X\mid Y) = \Hsf(Y) - \Hsf(Y\mid X)
\]
and the \emph{conditional mutual information between $X$ and $Y$ given $Z$} is defined as
\[
    \I(X;Y\mid Z) \coloneqq \Hsf(X\mid Z) - \Hsf(X\mid Y,Z) = \Hsf(Y\mid Z) - \Hsf(Y\mid X,Z).
\]
\end{definition}

\begin{fact}[Chain Rule]
\label{fact:chain-rule}
Let $X_1, X_2, Y, Z$ be discrete random variables. Then,
\[
    \I(X_1, X_2; Y\mid Z) = \I(X_1; Y\mid Z) + \I(X_2; Y\mid X_1, Z)
\]
\end{fact}

\begin{fact}
\label{fact:cond-entropy}
Let $X,Y$ be discrete random variables. Then, $\Hsf(X) \geq \Hsf(X\mid Y)$, with equality when $X$ and $Y$ are independent.
\end{fact}

\begin{lemma}[Information cost decomposition (Lemma 5.1, \cite{BJKS2004})]
\label{lem:ic-decomposition}
Let $\Pi$ be a protocol over $\mathcal L^n$ for some $\mathcal L\subseteq\mathcal X\times \mathcal Y$. Let $\zeta$ be a mixture of product distributions on $\mathcal L\times \mathcal D$, let $\eta = \zeta^n$, and suppose $((X,Y),D)\sim \eta$. Then, $\I(X,Y;\Pi(X,Y)\mid D) \geq \sum_{j=1}^n \I(X^j, Y^j; \Pi(X,Y)\mid D)$. 
\end{lemma}

\subsection{Total Variation Distance Lemma}

We need the following total variation distance calculation:

\begin{lemma}[Total variation distance bound]
\label{lem:tvd}
Let $\mu$ be a distribution over a finite alphabet $Q$ and let $\mathcal D \coloneqq \mu^d$. Let $\mathcal D'$ be the same distribution, except a uniformly random index $i\sim[d]$ is set to some $q^*\in Q$. Then, 
\[
\TV(\mathcal D, \mathcal D') \leq \sqrt{\frac{1-\mu(q^*)}{\mu(q^*)}}\frac1{\sqrt{d}}
\]
\end{lemma}
\begin{proof}
For any $x\in Q^d$ and $q\in Q$, let
\[
    s_{q}(x) = \abs{\braces*{q\in Q : x_j = q}}
\]
denote the number of coordinates $j\in[d]$ such that $x_j = q$. Then, we have that
\begin{align*}
    \mathcal D(x) &= \prod_{q\in Q}\mu(q)^{s_q(x)} \\
    \mathcal D'(x) &= \sum_{x_j = q^*}\Pr(x \mid I = j)\Pr(I = j) = \frac{s_{q^*}(x)}{d}\frac1{\mu(q^*)}\prod_{q\in Q}\mu(q)^{s_q(x)} .
\end{align*}
Then,
\begin{align*}
    \TV(\mathcal D, \mathcal D') &= \sum_{x\in Q^d}\abs{\mathcal D(x) - \mathcal D'(x)} \\
    &= \frac1{\mu(q^*)}\sum_{x\in Q^d}\mathcal D(x)\abs{\mu(q^*) - \frac{s_{q^*}(x,y)}{d}} \\
    &= \frac1{\mu(q^*) d}\sum_{x\in Q^d}\mathcal D(x)\abs{s_{q^*}(x) - \mu(q^*)d} \\
    &= \frac1{\mu(q^*)d}\expec{x\sim \mathcal D}{\abs{s_{q^*}(x) - \mu(q^*)d}} \\
    &\leq \frac1{\mu(q^*)d}\sqrt{\vars{x\sim \mathcal D}{s_{q^*}(x)}} \\
    &= \frac1{\mu(q^*)d}\sqrt{d\cdot \mu(q^*)(1-\mu(q^*))} \\
    &= \sqrt{\frac{1-\mu(q^*)}{\mu(q^*)}}\frac1{\sqrt{d}}.
\end{align*}
\end{proof}

\section{Proof of \texorpdfstring{$\tilde\Omega(n/\epsilon)$}{Omega(n/eps)} Lower Bound for \texorpdfstring{$k$}{k}-Means Clustering}
\label{sec:n/eps-full-proof}

\subsection{Hardness Lemma for Assignment to Centers}

In this section, we show information complexity lower bounds for a multi-player communication game based on a point assignment problem, when the input instance to the assignment problem is given by the sum $Z = \sum_{l=1}^t X^{(l)}\in\mathbb R^d$ of vectors $X^{(1)}, X^{(2)}, \dots, X^{(t)}\in\mathbb R^d$, each held by one of $t$ players, and we must assign $Z$ to the closest center $c^j\in\mathbb R^d$ for $j\in[k]$.

\subsubsection{Assignment of a Single Point}

We start by studying the problem of assigning a single point to a set of centers, as well as a hard random instance for this problem. Our instance is based on the information theoretic approach to the set disjointness problem and its $t$-bit generalization due to \cite{BJKS2004}. We define the point assignment problem as follows:

\begin{definition}[Point assignment problem]
\label{def:point-assignment}
Let $X^{(i)}\in\{0,1\}^d$ be binary vectors for $i\in[t]$ such that $Z = \sum_{i=1}^t X^{(i)}$ has at most one entry $j\in[d]$ such that $Z_j > 1$. We say that a randomized protocol $\Pi(X^{(1)}, X^{(2)}, \dots, X^{(t)})$ solves the point assignment problem with probability at least $1-\delta$ if for any $X^{(i)}$, $\Pi(X^{(1)}, X^{(2)}, \dots, X^{(t)})$ outputs some $e_j\in[d]$ such that $Z_j = t$ if such a $j\in[d]$ exists and any $e_l$ for $l\in[d]$ otherwise, with probability at least $1-\delta$.
\end{definition}

The hard instance that we study for the point assignment problems is generated as follows. For each of the $d$ coordinates, with probability $1/2$, we set the $j$th coordinates of the $t$ players' vectors to all zeros, and with probability $1/2$, we set the $j$th coordinate of a uniformly random player to $1$, and everyone else's $j$th coordinate to $0$. Finally, we select a uniformly random coordinate $j\in[d]$, and set the $j$th coordinate to $1$ for every player with probability $1-\alpha$ and $0$ for every player with probability $\alpha$. The formal definition is given in Definition \ref{def:bjks-instance}:

\begin{definition}[Hard instance for point assignment]\label{def:bjks-instance}
We define a distribution over $t$ random bit vectors in $d$ dimensions $\{X^{(i)}\}_{i=1}^t$ as follows. Let $B = \{B^j\}_{j=1}^d \sim [t]^d$, and let $I\sim[d]$ be a uniformly random index. Then for $j = I$, we draw the $j$th coordinates $\{X^{(i)}_j\}_{i=1}^t$ as
\[
    C\sim\begin{cases}
    (1,1, \dots, 1) & \text{w.p. $1-\alpha$} \\
    (0,0,\dots, 0) & \text{w.p. $\alpha$}
    \end{cases}
\]
and for $j \neq I$, we draw the $t$ values $\{X^{(i)}_j\}_{i=1}^t$ on the $j$th coordinate of each $X^{(i)}$ uniformly from $\{0, e_l\}$ where $l = B^j$. Let $\zeta$ denote the distribution over $(\{X^{(i)}\}_{i=1}^t, (I,B,C))$ on a single coordinate. We also denote by $Z$ the sum $Z = \sum_{i=1}^t X^{(i)}\in\mathbb R^d$.
\end{definition}

Throughout this section, we assume that $\Pi$ is a randomized protocol that solves the point assignment problem with probability at least $1-\delta$. We now derive information complexity lower bounds for this problem, on the input instance of Definition \ref{def:bjks-instance}. We refer to Appendix \ref{sec:info-theory} for standard preliminaries for information theory.

A crucial definition for the proof of the set disjointness information complexity lower bound of \cite{BJKS2004}, as well as our point assignment lower bound, is the following:

\begin{definition}[Conditional information complexity (Definition 4.5, \cite{BJKS2004})]
The $\delta$-error conditional information complexity of a function $f:\mathcal X^t\to\mathcal Y$ with respect to a distribution $\zeta$, denoted by $\CIC_{\zeta,\delta}(f)$, is defined as the smallest value of $\I(\{X^{(l)}\}_{l=1}^t; \Pi(\{X^{(l)}\}_{l=1}^t)\mid T)$ over the input distribution $(\{X^{(l)}\}_{l=1}^t, T)\sim \zeta$  for any $\delta$-error protocol $\Pi$ for $f$, that is, a protocol $\Pi$ which errs with probability at most $\delta$ on any input.
\end{definition}

We first show in Lemma \ref{lem:reduction-lemma} that for $\Omega(d)$ coordinates $j\in[d]$, the $j$th coordinate must reveal $\Omega(1/t^2)$ bits of information, by lower bounding the information cost on the $j$th coordinate by the conditional information complexity of the $t$-bit $\AND$ problem, that is, $\AND_t(x^{(1)}, x^{(2)}, x^{(t)}) \coloneqq \bigwedge_{l=1}^t x^{(l)}$. This conditional information complexity term is bounded by $\Omega(1/t^2)$ by Theorem 7.2 of \cite{BJKS2004}. As done in \cite{BJKS2004}, the only valid inputs to the $\AND_t$ problem that we consider are the all $0$ vector, the all $1$ vector, and the $t$ standard basis vectors $e_l\in\{0,1\}^t$ for $l\in[t]$.

\begin{lemma}[Reduction lemma]
\label{lem:reduction-lemma}
For at least $d/3$ coordinates $j\in[d]$,
\[
    \I(\{X^{(l)}_j\}_{l=1}^t; \Pi(\{X^{(l)}\}_{l=1}^t)\mid I, B, C) \geq \frac{\alpha}2\CIC_{\zeta,\delta'}(\AND_t)
\]
for $\delta' \coloneqq 4(3\delta+2/\sqrt{d-1}) + (3/d + \sqrt{2t}/\sqrt d)$, where $\zeta$ is the distribution defined in Definition \ref{def:bjks-instance}.
\end{lemma}
\begin{proof}
Our proof roughly follows Lemma 5.2 of \cite{BJKS2004}.

\paragraph{Identifying $d/3$ good coordinates.}

We first show that for a large number of coordinates $j\in[d]$, the protocol $\Pi$ is correct for the $\AND_t$ problem when restricted to the $j$th coordinate, that is, $\Pi$ outputs coordinate $j$ when $I = j$ and $C = (1, 1, \dots, 1)$, while $\Pi$ outputs a coordinate other than $j$ when $C \neq (1, 1, \dots, 1)$. 

For $j\in[d]$, let $\delta(j)$ denote the failure probability of the protocol $\Pi$ over the input distribution of Definition \ref{def:bjks-instance}, conditioned on $I = j$. By averaging, we have that $\delta(j) \leq 3\delta$ for at least $(2/3)d$ coordinates $j\in[d]$. Next, for $j\in[d]$, let $p(j)$ denote the probability that the protocol $\Pi$ outputs the standard basis vector $e_j$, conditioned on $I = j$ and $C \neq (1, 1, \dots, 1)$. First, if the input distribution is just the product distribution with each coordinate drawn as $\{X^{(i)}_j\}_{i=1}^t$ for $(\{X^{(i)}_j\}_{i=1}^t, D^j)\sim \zeta$, then note that at least $(2/3)d$ coordinates $j\in[d]$ will have $e_j$ output with probability at most $3/d$. Now if instead we uniformly draw $I\sim[d]$ and set $\{X^{(i)}_I\}_{i=1}^t = C$ for some $C \neq (1, 1, \dots, 1)$, then the total variation distance between this distribution and the product distribution is at most $\sqrt{2t}/\sqrt d$ by a total variation distance calculation carried out in Lemma \ref{lem:tvd}. Thus, $p(j) \leq 3/d + \sqrt{2t}/\sqrt d$ for these $(2/3)d$ coordinates $j$. Now by a union bound, there are at least $d/3$ coordinates such that $\delta(j) \leq 3\delta$ and $p(j) \leq 3/d + \sqrt{2t}/\sqrt d$. We will show the information complexity lower bound on these coordinates. From this point forth in this proof, we fix $j$ to be such a coordinate.

\paragraph{Reduction lemma.}

Note that for any $j\in[d]$,
\begin{align*}
    &\I(\{X^{(l)}_j\}_{l=1}^t; \Pi(\{X^{(l)}\}_{l=1}^t)\mid I, B, C) \\
    =~&\expec{i\sim I, b^{-j}\sim B^{-j}}{\I(\{X^{(l)}_j\}_{l=1}^t; \Pi(\{X^{(l)}\}_{l=1}^t)\mid I=i, B^j, B^{-j}=b^{-j}, C)} \\
    \geq~& \alpha \expec{i\sim I, b^{-j}\sim B^{-j}}{\I(\{X^{(l)}_j\}_{l=1}^t; \Pi(\{X^{(l)}\}_{l=1}^t)\mid I=i, B^j, B^{-j}=b^{-j}, C=(0,0,\dots,0))}
\end{align*}
where the last inequality is true since $C = (0, 0, \dots, 0)$ with probability $\alpha$.

Next, for each pair $(i, b^{-j})$, we construct a protocol $\Pi_{i,b}$ for a single copy of the AND problem with conditional information complexity loss exactly equal to 
\[
    \I(\{X^{(l)}_j\}_{l=1}^t; \Pi(\{X^{(l)}_j\}_{l=1}^t)\mid I=i, B^j, B^{-j}=b^{-j}, C=(0,0,\dots,0)).
\]
Let $\{x^{(l)}\}_{l=1}^t$ be a single copy of the $t$-bit AND problem. First note that conditioned on $I$, $B^{-j}$, and $C$, the hard instance of Definition \ref{def:bjks-instance} is a product distribution, that is, the $t$ players can generate their inputs independently for all coordinates except $j$. Then, the $t$ players generate such an input instance according to $I = i$, $B^{-j} = b^{-j}$, and $C=(0,0,\dots,0)$, and then replaces the $j$th input by $\{x^{(l)}\}_{l=1}^t$. The $t$ players then simulate the original protocol $\Pi$ with this input, and outputs $1$ as the answer to the AND problem if $\Pi$ assigns the $j$th standard basis vector to $Z = \sum_{l=1}^t X^{(l)}$, and $0$ otherwise.

Note that if $\{x^{(l)}\}_{l=1}^t$ is drawn according to the distribution of $C$ in Definition \ref{def:bjks-instance} and the index $i$ on which to plant $C = (0, 0, \dots, 0)$ is drawn uniformly randomly, then by Lemma \ref{lem:tvd}, the total variation distance between $\mathcal D$ conditioned on $I = j$ and the simulated input distribution $\mathcal D'$ is at most $2/\sqrt{d-1}$ (note that there are two different ``$I$''s here, one for the original problem instance where we are setting the random coordinate $I = i$ to be all zeros, and one for the fixed coordinate $I = j$ to be the planted input $\{x^{(l)}\}_{l=1}^t$ in the simulated instance). Then, letting $S(\{X^{(l)}\}_{l=1}^t)$ be the event that the protocol $\Pi$ is successful on input $\{X^{(l)}\}_{l=1}^t$, we have that
\begin{align*}
    &\Pr_{\{X^{(l)}\}_{l=1}^t\sim\mathcal D'}[S(\{X^{(l)}\}_{l=1}^t)] \\
    \geq~&\Pr_{\{X^{(l)}\}_{l=1}^t\sim\mathcal D}[S(\{X^{(l)}\}_{l=1}^t)] - \abs{\Pr_{\{X^{(l)}\}_{l=1}^t\sim\mathcal D}[S(\{X^{(l)}\}_{l=1}^t)] - \Pr_{\{X^{(l)}\}_{l=1}^t\sim\mathcal D'}[S(\{X^{(l)}\}_{l=1}^t)]} \\
    \geq~& \Pr_{\{X^{(l)}\}_{l=1}^t\sim\mathcal D}[S(\{X^{(l)}\}_{l=1}^t)] - \TV(\mathcal D,\mathcal D') \\
    \geq ~& 1 - 3\delta - \frac2{\sqrt{d-1}}.
\end{align*}
Thus, $\Pi$ is successful with probability at least $1 - 3\delta - 2/\sqrt{d-1}$ under $\mathcal D'$. Then by averaging, we have that for at least $d/2$ choices of $I = i$, the $\Pi$ is successful with probability at least $1 - 2(3\delta + 2/\sqrt{d-1})$ conditioned on the choice of $I = i$. 

Next, we bound the correctness probability of the protocol $\Pi_{i,b}$ for the $\AND_t$ problem, for the set of $d/2$ choices of $i$ as defined above. First, note that on this instance, if $\{x^{(l)}\}_{l=1}^t = (1,1,\dots,1)$, then $\Pi$ is correct if and only if it assigns $Z$ to $e_j$, since $Z_j = t$ whereas $Z_l \leq 1$ for every other $l\in[d]$. Since $\Pi$ must be correct with probability at least $1 - 2(3\delta+2/\sqrt{d-1})$ overall, it is correct with probability at least $1 - 4(3\delta+2/\sqrt{d-1})$ conditioned on $\{x^{(l)}\}_{l=1}^t = (1,1,\dots,1)$. On the other hand, if $\{x^{(l)}\}_{l=1}^t \neq (1,1,\dots,1)$, then by our condition on the coordinate $j$, $\Pi$ assigns $e_j$ to $Z$ with probability at most $3/d + \sqrt{2t}/\sqrt d$. Thus, for these inputs, $\Pi_{i,b}$ is correct with probability at least $1 - (3/d + \sqrt{2t} / \sqrt d)$. Thus, overall, $\Pi_{i,b}$ is correct with probability at least $1 - 4(3\delta+2/\sqrt{d-1}) - (3/d + \sqrt{2t}/\sqrt d) = 1 - \delta'$ on any input.

Finally, let $(\{X'^{(l)}\}_{l=1}^t,B')\sim\zeta$. Then, note that the joint distribution of $(\{X'^{(l)}\}_{l=1}^t, B', \Pi_{i,b})$ is exactly the same as the joint distribution of $(\{X_j^{(l)}\}_{l=1}^t, B^j, \Pi(\{X^{(l)}\}_{l=1}^t))$, conditioned on $I = i, B^{-j}=b^{-j}, C=(0,0,\dots,0)$. Thus, this shows that
\begin{align*}
    &\I(\{X^{(l)}_j\}_{l=1}^t; \Pi(\{X^{(l)}\}_{l=1}^t)\mid I=i, B^j, B^{-j}=b^{-j}, C=(0,0,\dots,0)) \\
    = ~&\I(\{X'^{(l)}\}_{l=1}^t; \Pi_{i,b}) \geq \CIC_{\zeta,\delta'}(\AND_t).
\end{align*}
Chaining together the previous inequalities yields the claimed result.
\end{proof}

Combining Lemma \ref{lem:reduction-lemma} with Lemma \ref{lem:ic-decomposition} yields the following:

\begin{lemma}
\label{lem:single-point-assign}
For $\delta \leq 1/50$ and $\sqrt{2t/d} \leq 1/20$, we have
\[
    \I(\{X^{(l)}\}_{l=1}^t;\Pi(\{X^{(l)}\}_{l=1}^t)\mid I, B, C) = \Omega(d/t^2).
\]
\end{lemma}
\begin{proof}
If $\delta\leq 1/50$, then $12\delta \leq 12/50 < 1/4$ so for large enough $d$, the $\delta'$ in Lemma \ref{lem:reduction-lemma} is at most $1/3$. In this case, $\CIC_{\zeta,\delta'}(\AND_t) = \Omega(1/t^2)$ by Theorem 7.2 of \cite{BJKS2004}, which, combined with Lemma \ref{lem:ic-decomposition}, yields the statement of the lemma.
\end{proof}

\subsubsection{Assignment of Multiple Points}

Next, we show by a direct sum argument that solving the assignment problem for $n$ points requires a protocol to reveal $\Omega(nd/t^2)$ bits of information.

\begin{lemma}
\label{lem:multiple-point-assign}
Let $\sqrt{2t/d} \leq 1/20$. Let $(\{X^{(l)}\}_{l=1}^t, (I, B, C)) = \{(\{X^{(i,l)}\}_{l=1}^t, (I^i, B^i, C^i))\}_{i=1}^n$ be drawn as $n$ i.i.d.\ from the hard distribution of Definition \ref{def:bjks-instance}. Suppose that a protocol $\Pi$ outputs a correct solution to the point assignment problem of Definition \ref{def:point-assignment} for least a $399/400$ fraction of points $\{X^{(i,l)}\}_{l=1}^t$ for $i\in[n]$, with probability at least $399/400$. Then,
\[
    \I(\{X^{(l)}\}_{l=1}^t; \Pi(\{X^{(l)}\}_{l=1}^t)\mid I, B, C) = \Omega(nd/t^2).
\]
\end{lemma}
\begin{proof}
Let $i\sim[n]$ be a uniformly random index. Then, by a union bound, the $i$th instance of the point assignment problem is solved correctly with probability at least $1 - 2/400 = 1 - 1/200$. Now for each fixed $i\in[n]$, let $\delta(i)$ be the probability that the $i$th instance is solved correctly. Then, over the randomness used by the protocol as well as $i\sim[n]$, we have that
\[
    \Pr_{i\sim[n]}\braces*{\text{$i$th instance is correct}} = \sum_{i=1}^n \frac1n \Pr\braces*{\text{$i$th instance is correct}} = \frac1n\sum_{i=1}^n 1 - \delta(i) \geq 1 - \frac{1}{200}
\]
so $\EE_{i\sim[n]}\delta(i) \leq 1/200$. Then for at least $n/2$ indices $i'\in[n]$, we have that $\delta(i') \leq 2/200 = 1/100$. We now claim that on these coordinates $i'\in[n]$, we have that
\[
    \I(\{X^{(i',l)}\}_{l=1}^t; \Pi(\{X^{(l)}\}_{l=1}^t)\mid I, B, C) = \Omega(d/t^2).
\]
Indeed, note that $\I(\{X^{(i',l)}\}_{l=1}^t; \Pi(\{X^{(l)}\}_{l=1}^t)\mid I, B, C)$ is the expectation of
\[
    \I(\{X^{(i',l)}\}_{l=1}^t; \Pi(\{X^{(l)}\}_{l=1}^t)\mid I^{i'}, B^{i'}, C^{i'}, I^{-i'}=i^{-i'}, B^{-i'}=b^{-i'}, C^{-i'}=c^{-i'})
\]
over $i^{-i'}\sim I^{-i'}, b^{-i'}\sim B^{-i'}, c^{-i'}\sim C^{-i'}$. Now for each fixing $i^{-i'}, b^{-i'}, c^{-i'}$, let $\delta(i^{-i'}, b^{-i'}, c^{-i'})$ that the $i'$th instance of the point assignment problem is correct given these fixings. Then by Markov's inequality, for at least half of the fixings, we have $\delta(i^{-i'}, b^{-i'}, c^{-i'}) \leq 2/100 = 1/50$. Note that each of these fixings corresponds to a protocol for solving the point assignment problem with probability at least $1-1/50$. Thus, we have by Lemma \ref{lem:single-point-assign} that
\[
    \I(\{X^{(i',l)}\}_{l=1}^t; \Pi(\{X^{(l)}\}_{l=1}^t)\mid I^{i'}, B^{i'}, C^{i'}, I^{-i'}=i^{-i'}, B^{-i'}=b^{-i'}, C^{-i'}=c^{-i'}) = \Omega(d/t^2)
\]
for these fixings. Since this event occurs with probability at least $1/2$, it follows that $\I(\{X^{(i',l)}\}_{l=1}^t; \Pi(\{X^{(l)}\}_{l=1}^t)\mid I, B, C) = \Omega(d/t^2)$ as well. 

Finally, by Lemma \ref{lem:ic-decomposition}, we have that
\begin{align*}
    \I(\{X^{(l)}\}_{l=1}^t; \Pi(\{X^{(l)}\}_{l=1}^t)\mid I, B, C) &\geq \sum_{i=1}^n \I(\{X^{(i,l)}\}_{l=1}^t; \Pi(\{X^{(l)}\}_{l=1}^t)\mid I, B, C) \\
    &\geq \frac{n}{2}\cdot\Omega(d/t^2) = \Omega(nd/t^2).
\end{align*}
\end{proof}

\subsection{Lower Bounds for Clustering in Row Insertion Streams}

Our first result is to show that an algorithm for computing a $(1+\epsilon)$-approximate nearly optimal $k$-means clustering on $n$ points for $k = d = \Theta(1/\epsilon)$ on row insertion streams requires $\Omega(n/\epsilon)$ bits of space.

For this result, we need a lower bound against any nearly optimal clustering, so we need to ``plant'' our desired centers in order to force the solution to look like standard basis vectors. This will allow us to use the clustering algorithm to solve the point assignment problem. In order to determine the number of points we need to plant the centers, we first need a lower bound on the cost of any clustering of random bits, which we show in the next section. 

\subsubsection{Cost Lower Bound on Random Points}

We first lower bound the cost of any clustering of the random points of the hard instance in Definition \ref{def:bjks-instance}. We start with a bound in expectation:

\begin{lemma}[Expectation bound for clustering random bits]
\label{lem:expectation-bound}
Fix a set of centers $c^1, c^2, \dots, c^k\in[0,1]^d$. Let $Z\in\{0,1\}^d$ be a vector of $d$ uniformly random bits. Then,
\[
    \expec{Z}{\min_{j=1}^k \norm{Z - c^j}_2^2} \geq \frac{d}{4} - \frac{\log(kd)+1}2.
\]
\end{lemma}
\begin{proof}
Let $\mu\coloneqq \EE[Z]$ (i.e., the vector with $1/2$ in every entry). Fix a specific center $c^j$ for $j\in[k]$. Then,
\[
    \norm{Z - c^j}_2^2 = \norm{Z - \mu}_2^2 + \norm{\mu - c^j}_2^2 + 2\angle{Z - \mu, \mu - c^j} = \frac{d}{4} + \norm{\mu - c^j}_2^2 + 2\angle{Z - \mu, \mu - c^j}
\]
By Hoeffding's inequality, we have
\[
    \Pr\braces*{\abs{\angle{Z - \mu, \mu - c^j}} \geq t\norm{\mu-c^j}_2} \leq 2\exp\parens*{-2t^2}
\]
so for $t = \sqrt{\log(kd)/2}$, this probability is at most $2/kd$. By a union bound over the $k$ choices of $j$, we have that
\[
    \Pr\braces*{\min_{j=1}^k \norm{Z - c^j}_2^2 \leq \frac{d}{4} + \norm{\mu - c^j}_2^2 - 2\sqrt{\log(kd)/2}\norm{\mu-c^j}_2} \leq \frac2d.
\]
Note that
\[
    \norm{\mu - c^j}_2^2 - 2\sqrt{\log(kd)/2}\norm{\mu-c^j}_2 = \parens*{\norm{\mu - c^j}_2 - \sqrt{\log(kd)/2}}^2 - \log(kd)/2 \geq -\log(kd)/2
\]
so
\[
    \Pr\braces*{\min_{j=1}^k \norm{Z - c^j}_2^2 \geq \frac{d}{4} - \frac{\log(kd)}2} \geq 1 - \frac2d.
\]
It follows that
\[
    \EE\bracks*{\min_{j=1}^k \norm{Z - c^j}_2^2} \geq \parens*{1-\frac2d}\parens*{ \frac{d}{4} - \frac{\log(kd)}2 } \geq \frac{d}{4} - \frac{\log(kd)+1}2
\]
\end{proof}

Lemma \ref{lem:expectation-bound} shows that when clustering random bits, we can only save approximately a $(1-1/\tilde\Theta(d))$ factor for any clustering compared to a single center, in expectation. Since all but one coordinate in the hard instance of Definition \ref{def:bjks-instance} are random bits, and the one coordinate can only decrease the cost by a factor of $(1-1/\tilde\Theta(d))$, any clustering into $k$ centers still has cost at least approximately $(1 - 1/\tilde\Theta(d))$ times the cost of a single center. 

The next lemma converts the result of Lemma \ref{lem:expectation-bound} into a high probability result about any clustering, via a net argument.

\begin{lemma}
\label{lem:cost-lb}
Let $\{Z^i\}_{i=1}^n$ be $n$ independent uniformly random bit vectors in $d$ dimensions. Suppose that $n \geq 16 d\log(d^{2d}/\delta) = 32 d^2\log(d/\delta)$. Then, with probability at least $1-\delta$, we have that
\[
    \min_{c^1, c^2, \dots, c^k\in[0,1]^d} \sum_{i=1}^n \min_{j=1}^k \norm{Z^i - c^j}_2^2 \geq n \parens*{\frac{d}{4} - \frac{\log(kd)+9}2}.
\]
\end{lemma}
\begin{proof}
Let $\{c^j\}_{j=1}^k$ and $\{c'^j\}_{j=1}^k$ be two sets of centers such that $\norm{c^j - c'^j}_2^2 \leq 1/d$. Then,
\begin{align*}
    \min_{j=1}^k \norm{Z^i - c^j}_2^2 &\leq \min_{j=1}^k \norm{Z^i - c'^j}_2^2 + \norm{c'^j - c^j}_2^2 + 2\norm{Z^i - c'^j}_2\norm{c'^j - c^j}_2 \\
    &\leq \min_{j=1}^k \norm{Z^i - c'^j}_2^2 + 3
\end{align*}
so if $\{c^j\}_{j=1}^k$ has high cost, then $\{c'^j\}_{j=1}^k$ must as well. 
We now consider a net $\mathcal N\subseteq[0,1]^d$ of size $d^{2d}$ such that for any $c\in[0,1]^d$, there exists $c'\in\mathcal N$ such that $\norm{c-c'}_2^2 \leq 1/d$. Now fix a set of centers $\{c^j\}_{j=1}^k \in\mathcal N^k$. By Lemma \ref{lem:expectation-bound}, we have that
\[
    \EE_Z[\min_{j=1}^k \norm{Z - c^j}_2^2] \geq \frac{d}{8}
\]
for sufficiently large $d$, so we have that
\[
    \Pr\braces*{\sum_{i=1}^n \min_{j=1}^k \norm{Z^i - c^j}_2^2 \leq (1-1/d) n \expec{Z}{\min_{j=1}^k \norm{Z - c^j}_2^2}} \leq \exp\parens*{-\frac{nd}{16 d^2}} \leq \frac{\delta}{d^{2d}}
\]
by Chernoff bounds. Then by a union bound, the same holds simultaneously for every $\{c^j\}_{j=1}^k \in\mathcal N^k$ with probability at least $1-\delta$.

Now for an arbitrary set of centers $c^1, c^2, \dots, c^k\in[0,1]^d$, there exists some $\{c'^j\}_{j=1}^k \in\mathcal N^k$ such that $\norm{c^j-c'^j}_2^2 \leq 1/d$ for every $j\in[k]$. Then,
\begin{align*}
    \sum_{i=1}^n \min_{j=1}^k \norm{Z^i - c^j}_2^2 &\geq \sum_{i=1}^n \parens*{\min_{j=1}^k \norm{Z^i - c'^j}_2^2 - 3} \\
    &\geq (1-1/d) n \expec{Z}{\min_{j=1}^k \norm{Z - c'^j}_2^2}-3n \\
    &\geq (1-1/d) n \parens*{\frac{d}{4} - \frac{\log(kd)+1}2} - 3n \\
    &\geq n \parens*{\frac{d}{4} - \frac{\log(kd)+9}2}.
\end{align*}
\end{proof}

\subsubsection{Upper Bound on a Nearly Optimal Cost}

We first upper bound the optimal cost of clustering by giving an explicit clustering construction, and upper bounding the cost. We define this clustering in Definition \ref{def:nearly-optimal-clustering}:

\begin{definition}[Nearly optimal clustering]
\label{def:nearly-optimal-clustering}
We define a clustering for points drawn from Definition \ref{def:bjks-instance}. Consider the variables $I$ and $C$ as defined in Definition \ref{def:bjks-instance}. If $C = (1, 1, \dots, 1)$ and $I = j$, then we assign the point to cluster $j$. On the other hand, if $C \neq (1, 1, \dots, 1)$ and $I = j$, then we assign the point to a uniformly random point $j'\in[d]\setminus\{j\}$ such that $X_{j'}^{(l)} = 1$ for some $l\in[t]$. If no such coordinate exists, we assign it to any cluster. Furthermore, we define the center $c^j$ by setting its $j'$th coordinate to be
\[
    c^j_{j'} = \begin{cases}
    \frac{t+1}2 & \text{if $j' = j$} \\
    \frac12 & \text{if $j'\neq j$}
    \end{cases}
\]
\end{definition}

The cost of this clustering is bounded in the following lemma:

\begin{lemma}
\label{lem:cost-ub}
Let $\{Z^i\}_{i=1}^n$ be drawn i.i.d.\ from the distribution of Definition \ref{def:bjks-instance}. Then, with probability at least $1 - (1/2)^{d-1}$, the clustering defined in Definition \ref{def:nearly-optimal-clustering} has cost at most $n(d+t^2-2t)/4$. 
\end{lemma}
\begin{proof}
Let $(\{X^{(i,l)}\}_{l=1}^t, (I^i, B^i, C^i))$ denote the $i$th element drawn from Definition \ref{def:bjks-instance}, for $i\in[n]$. We handle the cost calculation by conditioning on the event that at least one nonzero coordinate is drawn on $[d]\setminus \{I^i\}$, since this occurs with probability at least $1 - (1/2)^{d-1}$. 

Fix a cluster $j\in[k]$. We will consider the distribution of points $\{X^{(i,l)}\}_{l=1}^t$, conditioned on the event that the point being clustered to cluster $j$ in the clustering of Definition \ref{def:nearly-optimal-clustering}. Note then that the $j$th coordinate comes from a point such that $I^i = j$ and $C^i = (1, 1, \dots, 1)$, or the $j$th coordinate comes from a point with $\sum_{l=1}^t X^{(i,l)} = 1$ and $I^i \neq j$ and $C^i = (0, 0, \dots, 0)$. In either case, the coordinates $[d]\setminus\{j\}$ are in $\{0,1\}$, and the $j$th coordinate is in $\{1,t\}$. Then for our defined center $c^j$, the squared cost is $(1/2)^2 = 1/4$ on $d-1$ coordinates and $((t-1)/2)^2 = (t-1)^2/4$ on one coordinate per point, for a total of $n\cdot ((d-1)/4 + (t-1)^2/4) = n(d+t^2-2t)/4$ as claimed.
\end{proof}

\subsubsection{Planting Centers}

With our nearly optimal clustering of Definition \ref{def:nearly-optimal-clustering} in mind, we now add copies of these centers into our instance in order to encourage the clustering algorithm to find this solution. Note that this increases the cost of any other clustering, without increasing the cost of this clustering.

\begin{lemma}
\label{lem:planted-centers}
Let $n \geq 32 d^2\log(d/\delta)$. Consider the input instance to $k$-means clustering given by $n$ random points drawn according to Definition \ref{def:bjks-instance}, together with 
\[
    \gamma \coloneqq \frac{400 t^2 n}{k}\parens*{\frac{\log(kd)+9}2 + \frac{t^2-2t}{4} + \frac{(d+t^2-2t)}{4d}} = O\parens*{\frac{t^2 n}{k}(\log(kd) + t^2)}
\]
copies of each center $c^j$ for $j\in[k]$ as defined in Definition \ref{def:nearly-optimal-clustering}. Furthermore, let $\{\hat c^j\}_{j=1}^k$ be centers achieving a $(1+1/d)$-nearly optimal solution to the $k$-means clustering instance. Then, $\norm{c^j-\hat c^j}_2^2 \leq 1/4$ for at least $(1-1/100 t^2)k$ of the centers $c^j$.
\end{lemma}
\begin{proof}
Recall that in Lemma \ref{lem:cost-lb}, we showed that any clustering of $n$ random points drawn from Definition \ref{def:bjks-instance} must have a cost of at least $nd/4 - n(\log(kd)+9)/2$ with probability at least $1 - \delta$. Then, with probability at least $1-(1/2)^{d-1}$, the value of the optimal solution is bounded above by $n(d+t^2-2t)/4$ by Lemma \ref{lem:cost-ub}, so we must have that
\[
    \gamma \sum_{j=1}^k \norm{c^j - \hat c^j}_2^2 + n \parens*{\frac{d}{4} - \frac{\log(kd)+9}2} \leq (1+1/d)\frac{n(d+t^2-2t)}{4}
\]
which implies that
\[
    \frac{1}{k}\sum_{j=1}^k \norm{c^j - \hat c^j}_2^2 \leq \frac1{400 t^2}
\]
by rearranging. By averaging, at least $(1-1/100t^2)k$ of the $k$ centers $j\in[k]$ satisfy $\norm{c^j-\hat c^j}_2^2 \leq 1/4$.
\end{proof}

Note that Lemma \ref{lem:planted-centers} only allows us to characterize the behavior of $(1-1/100t^2)k$ many cluster centers, which still allows for the possibility that the remaining $k/100t^2$ centers are able to fit many points with low cost. The following lemmas show that this cannot happen.

\begin{lemma}
\label{lem:few-points-per-center}
Consider a set of $k'$ centers $\hat c^j\in\mathbb R^d$ for $j\in[k']$. Let $\{Z^i\}_{i=1}^{n'}$ be $n'\geq M$ points such that $Z^i$ takes the value $t$ on coordinate $l^i\in[d]$, and furthermore, we have $\abs{\braces{i\in[n'] : Z_l^i = t}} \leq M$ for any $l\in[d]$. Then, the cost of any clustering of these $n'$ points with $k'$ clusters is at least
\[
    n'\frac{d}{4} - n\frac{\log((k'+1)d)+9}{2} + \frac45 t^2\parens*{n' - 10k'\cdot M}
\]
\end{lemma} 
\begin{proof}
We first lower bound the cost of the $k'$ centers by a ``random'' part of the cost and the ``spike'' part of the cost. For each $j\in[k']$, define the center $\bar c^j$ which is the center $\hat c^j$ with all entries greater than $1$ set to $1$. 

Suppose that $Z^i$ is a point with some coordinate $l\in[d]$ such that $Z_l^i = t$. Note then that on the $l$th coordinate, we have that
\begin{align*}
    (Z_l^i - \hat c_l^j)^2 &\geq (Z_l^i-\hat c_l^j)^2 + (b^i-\bar c_l^j)^2 - 1
\end{align*}
for some random bit $b^i\sim\{0,1\}$. For all other coordinates $l\in[d]$, if $\hat c_l^j > 1$, then we lower bound the cost on the $l$th coordinate by
\begin{align*}
    (Z_l^i-\hat c_l^j)^2 \geq (Z_l^i-1)^2 + (1-\hat c_l^j)^2 = (Z_l^i-\bar c_l^j)^2 + (1-\hat c_l^j)^2
\end{align*}
while if $\hat c_l^j \leq 1$, then we simply write the cost as $(Z_l^i-\hat c_l^j)^2 = (Z_l^i - \bar c_l^j)^2$. Note then that the cost lower bounds derived above can be grouped into a cost corresponding to a clustering cost of random bit vectors with centers $\bar c^j\in\mathbb R^d$, and everything else. 

We will first lower bound the latter costs. Note that these costs are given by $(t-\hat c_l^j)^2 - 1$ for the coordinate $l\in[d]$ such that $Z_l^i = t$ and $(\hat c_l^j-1)^2$ for the coordinates $l\in[d]$ such that $\hat c_l^j > 1$. In fact, we can note that this is just one less than the $\ell_2$ distance between $\hat c^j$ and the vector $(1, 1, \dots, 1, t, 1, \dots, 1)$, i.e., the all ones vector with $t$ in the $l$th position, since we can WLOG threshold all entries of $\hat c^j$ less than $1$ to be exactly $1$. Note that this cost is minimized when there are $n'/M$ different indices $l\in[d]$, each which has $\abs{\braces{i\in[n']:Z_l^i=t}} = M$, and when all vectors $Z^i$ with the same coordinate $l$ for $Z_l^i = t$ are clustered to the same center (see, e.g., \cite{FWY2019}). For each $l\in[d]$, denote by $G^{(l)}$ the set $\braces{i\in[n']:Z_l^i=t}$. Then, there are at most $10k'$ indices $l\in[d]$ that belong to clusters consisting of at most $10$ groups $G^{(l)}$. All other indices $l\in[d]$ belong to clusters that consist of at least $10$ groups $G^{(l)}$, and thus the center of this cluster has coordinates with magnitude at most $t/10$. Thus, for at least $n' - 10k'\cdot M$ points, the cost is at least $(t-t/10)^2 = (9/10)^2 t^2 \geq (4/5)t^2$. 

Next, we lower bound the cost of clustering the random bit vectors by $\bar c^j$. By Lemma \ref{lem:cost-lb}, the total cost of any clustering of $n$ random points with $k'+1$ clusters must be at least
\[
    n\parens*{\frac{d}{4} - \frac{\log((k'+1)d)+9}{2}}.
\]
One way to cluster these $n$ random points is to first cluster $n'$ points using $k'$ clusters, and then cluster all the remaining $n-n'$ points with the fixed center given by the vector with all $1/2$s, which gives a cost of $d/4$ for any point. Then by the above cost lower bound, it follows that the cost of the clustering of the $n'$ points using the $k'$ clusters must be at least
\[
    n\parens*{\frac{d}{4} - \frac{\log((k'+1)d)+9}{2}} - (n-n')\frac{d}{4} = n'\frac{d}{4} - n\frac{\log((k'+1)d)+9}{2}.
\]

\end{proof}

\subsubsection{Reduction from Point Assignment}

Finally, we obtain an information complexity lower bound for the $k$-means clustering problem, by a reduction from the point assignment problem of Lemma \ref{lem:multiple-point-assign}. 

\begin{theorem}
\label{thm:k-means-lb}
Let $t = \max\{2000, 80\sqrt{\log(kd)+10}+2\}$. Let $\{Z^i\}_{i=1}^n$ be drawn i.i.d.\ from the distribution of Definition \ref{def:bjks-instance}, with $\alpha = 1/100t^2$. Consider the input instance given by these points, together with the planted centers as specified in Lemma \ref{lem:planted-centers}. Suppose that $\hat c^j\in\mathbb R^d$ for $j\in[k]$ are centers that achieve a $(1+\epsilon)$ approximation, for $\epsilon = (\log(kd)+10)/(d + (t-1)^2) = \tilde O(1/d)$. Suppose that we assign $e_l$ to $Z^i$ whenever $Z^i$ is clustered to the center $\hat c^j$ that has largest entry in the $l$th coordinate for $l\in[d]$. Then, this solves the point assignment problem (Definition \ref{def:point-assignment}) for at least $(399/400)n$ of the $Z^i$ for $i\in[n]$. Hence, solving $k$-means clustering up to $(1+\epsilon)$ accuracy on this instance requires $\Omega(nd/t^2) = \tilde \Omega(nd) = \tilde\Omega(n/\epsilon)$ bits of communication.
\end{theorem}
\begin{proof}
Let $\{\hat c^j\}_{j=1}^k$ be a clustering achieving a $(1+\epsilon)$ approximation. We will show that we must have at most $n/400$ incorrect classifications of the points $Z^i$.

We first introduce some notation. For each $j\in[k]$, we let $G^{(j)}\subseteq[n]$ denote the subset of points $i\in[n]$ such that $Z^i_j = t$, and we let $G^{(0)}\coloneqq [n] \setminus \bigcup_{j\in[k]} G^{(j)}$ denote the set of points such that $\norm{Z^i}_\infty \leq 1$. Note then that $G^{(0)}$ corresponds to the set of points with $C = (0, 0, \dots, 0)$ for $C$ defined in Definition \ref{def:bjks-instance}, and thus has size $\EE\abs{G^{(0)}} = \alpha n$ in expectation and size $\Theta(\alpha n)$ with probability at least $1-\delta$ by Chernoff bounds. We will also define $\bar c^j$ for each $j\in[k]$ to be the center $\hat c^j$ with any entry larger than $1$ set to be equal to $1$.

By Lemma \ref{lem:planted-centers}, there is a subset $S\subseteq[k]$ of size at least $\abs{S} \geq (1-1/100 t^2)k$ such that $\norm{c^j-\hat c^j}_2^2 \leq 1/4$. We make use of this fact later, and first bound the cost of points that can be clustered by the remaining at most $k' = \abs{[k]\setminus S} \leq k/100t^2$ centers. Note that by Chernoff bounds and a union bound, we have that $\abs{\braces{i\in[n] : Z_l^i = t}} \leq 2n/k$ for every $l\in[n]$. Then by Lemma \ref{lem:few-points-per-center}, if there are $n'$ points clustered by these $k'$ centers, then the cost is at least
\begin{equation}\label{eq:bad-center-cost}
    n'\frac{d}{4} - n\frac{\log(kd)+9}{2} + \frac45 t^2\parens*{n' - 10\frac{k}{100t^2}\frac{2n}{k}} \geq n'\parens*{\frac{d}{4} + \frac45 t^2} - n\frac{\log(kd)+10}{2}
\end{equation}

Now let $j\in S$. We will bound the cost of the points $Z^i\in G^{(j)}$, as a function of the number of points that are clustered to some center $\hat c^{j'}$ for $j'\neq j$. Let $Z^i$ be a point clustered to some center $\hat c^{j'}$ for $j'\neq j$ and $j'\in S$ (recall that we have already handled the cost of clustering points to centers outside of $S$). Then, the cost on the $j$th coordinate is bounded below by
\[
    (Z_{j}^i - \hat c_j^{j'}) \geq \parens*{t-\frac{1}{2} - \norm{c^{j'}-\hat c^{j'}}_\infty}^2 \geq (t-1)^2.
\]
On the other hand, if the assigned center is correct, i.e. $j' = j$, then the cost lower bound on the $j$th coordinate is
\[
    (Z_{j}^i - \hat c_j^{j'}) \geq \parens*{t-\frac{t+1}{2} - \norm{c^{j'}-\hat c^{j'}}_\infty}^2 \geq (t-2)^2/4.
\]
Thus, each incorrectly classified point pay an additional cost $(t-1)^2 - (t-2)^2/4 \geq (t-2)^2/2$ on the $j$th coordinate. We will later lower bound the cost of the rest of the coordinates via Lemma \ref{lem:cost-lb}. 

In the last remaining cases of $i\in G^{(0)}$ and $i\in G^{(j)}$ for $j\notin S$, we will only be able to lower bound the cost by the cost of the random coordinates via Lemma \ref{lem:cost-lb}, but not by the additional $(t-2)^2/4$ term on the $j$th coordinate. This will be fine, as there are only roughly $n/t^2$ such points, since $\abs{G^{(0)}} \leq 2\alpha n = n/50t^2$ and $\abs{[k]\setminus S} \leq (1/100t^2) k$ so
\[
    \abs{\bigcup_{j\in [k]\setminus S}G^{(j)}} \leq \frac{k}{100t^2}\frac{2n}{k} \leq \frac{n}{50 t^2}.
\]
Thus, at least $n - n' - (n/50t^2 + n/50t^2)$ points will incur a cost of $(t-2)^2/4$, for a cost contribution of
\[
    \frac{(t-2)^2}{4}\parens*{n - n' - (n/50t^2 + n/50t^2)} = (n-n')\frac{(t-2)^2}{4} - \frac{n}{100}
\]

Finally, we bring all the above calculations together. Suppose that there are $b$ points $Z^i$ that belong to $G^{(j)}$ for some $j\in S$, but are clustered to some other $\hat c^{j'}$ for $j'\in S$. First, the cost of the points that are clustered to some center not in $S$ is given in \eqref{eq:bad-center-cost}. Next, the cost of clustering the random coordinates of all other points is similarly bounded below by Lemma \ref{lem:cost-lb} by
\[
    (n-n')\frac{d}{4} - n\frac{\log(kd)+9}{2}.
\]
Thus, altogether, the cost is bounded below by
\begin{align*}
    &b\frac{(t-2)^2}{2} + (n-n')\parens*{\frac{d}{4} + \frac{(t-2)^2}{4}} + n'\parens*{\frac{d}{4} + \frac45 t^2} - n(\log(kd)+10) \\
    \geq ~&b\frac{(t-2)^2}{2} + \frac{nd}{4} + n\frac{(t-2)^2}{4} + n'\frac{(t-2)^2}{2} - n(\log(kd)+10)
\end{align*}
Then, if $b$ or $n'$ are greater than $n/800$, then this cost is at least
\[
    \frac{n}{800}\frac{(t-2)^2}{2} + \frac{nd}{4} + n\frac{(t-2)^2}{4} - n(\log(kd)+10)
\]
For $t\geq 2000$, we have that
\[
    \frac12\frac{n}{800}\frac{(t-2)^2}{2} \geq \frac{n}{4}\cdot 2t
\]
and for $t\geq 80\sqrt{\log(kd)+10}+2$, we have that
\[
    \frac12\frac{n}{800}\frac{(t-2)^2}{2} \geq 2n(\log(kd)+10)
\]
and thus if both of these hold, then the cost is at least
\[
    \frac{nd}{4} + n\frac{(t-1)^2}{4} + n(\log(kd)+10).
\]
Thus, by our choice of $\epsilon$, this fails to be a $(1+\epsilon)$-approximate solution, and thus we must have that $b$ and $n'$ are both at most $n/800$. Thus, the algorithm can incorrectly classify at most $n/400$ points.
\end{proof}

\section{Missing Proofs from Section \ref{sec:lower-bounds}}

\subsection{Proof of Theorem \ref{thm:strong-lb}}
\label{sec:strong-lb-proof}

\begin{proof}[Proof of Theorem \ref{thm:strong-lb}]
Let $d = 2\ceil{1/\epsilon^2}$ and let $X = \{X^i\}_{i=1}^n\subseteq\{0,1\}^{d/2}$ be a collection of $n$ uniformly random bit vectors, each with $d/2$ coordinates. Then for each $i\in[n]$, we form a vector $a^i\in\mathbb R^d$ by setting the $(2j-1)$th and $2j$th coordinates to be
\[
    (a^i_{2j-1}, a^i_{2j}) = \begin{dcases}
    (0, 1) & \text{if $X_j^i = 0$} \\
    (1, 0) & \text{if $X_j^i = 1$}
    \end{dcases}
\]

Fix any $j\in[d/2]$, and suppose that we query the cost of two centers given by the vectors $c^1 = \sqrt d \cdot e_{2j-1}$ and $c^2 = \sqrt d \cdot e_{2j}$. Then, the center cost query data structure must output a partition $C^1, C^2\subseteq[n]$ such that
\[
    \sum_{i\in C^1} \norm{a^i - c^1}_2^2 + \sum_{i\in C^2} \norm{a^i - c^2}_2^2 \leq (1+\epsilon/15)\cost(c^1, c^2).
\]
We claim that the partition must assign all but at most $n/10$ of the $a^i$ to its closest center. Note that this implies the theorem. Indeed, given the center cost query data structure $M$, we can reconstruct a bits $X'$ which agrees with $X$ on all but at most $(n/10)(d/2) = nd/20$ bits, so
\begin{align*}
    \Hsf(M) &\geq \Hsf(M) - \Hsf(M\mid X) \\
    &= \I(M; X) \\
    &\geq \I(X'; X) && \text{data processing inequality} \\
    &= \Hsf(X) - \Hsf(X\mid X') \\
    &\geq \frac{nd}{2} - \frac{nd}{20} = \Omega(nd).
\end{align*}
Then, $M$ must use at least $\Omega(nd)$ bits to describe, since the number of bits of a message upper bounds the entropy of a random variable.

Note first that the cost of this query on any vector is at least
\[
    (\sqrt d - 1)^2 \geq (1 - 1/\sqrt d)^2 d \geq (1-2/\sqrt d) d \geq d/2
\]
and at most
\[
    \norm{a^i - c^1}_2^2 \leq 2\norm{a^i}_2^2 + 2\norm{c^1}_2^2 = 3d.
\]
Thus, the total error that the partition can incur is at most
\[
    \sum_{i\in C^1} \norm{a^i - c^1}_2^2 + \sum_{i\in C^2} \norm{a^i - c^2}_2^2 - \cost(c^1, c^2)\leq \epsilon \cost(c^1, c^2) \leq 3\cdot \frac{\epsilon}{15} nd \leq \frac{\sqrt d}{5} n
\]
By averaging over the $n$ vectors, there can be at most $n/10$ indices $i\in[n]$ such that $a^i$ is assigned to a cluster with center $c\in\{c^1, c^2\}$ with
\[
    \norm{a^i-c}_2^2 - \min\braces*{\norm{a^i-c^1}_2^2, \norm{a^i-c^2}_2^2} \geq 2\sqrt d
\]

Now consider a single vector $a^i$, and say that $(a_{2j-1}^i, a_{2j}^i) = (0, 1)$. Note then that the difference between the cost of assigning this vector to $c^1$ versus the cost of assigning this vector to $c^2$ is at least
\[
    (\sqrt d)^2 + 1^2 - (\sqrt d - 1)^2 \geq 2 \sqrt d.
\]
Thus, there are at most $n/10$ vectors that can be assigned to the incorrect center.
\end{proof} 

\subsection{Proof of Theorem \ref{thm:lb-cost}}
\label{sec:lb-cost:proof}

\begin{proof}[Proof of Theorem \ref{thm:lb-cost}]
The proof is by a reduction from set disjointness \cite{razborov1990distributional}. Suppose that Alice and Bob are two players who hold an instance of set disjointness, that is, Alice has a subset $A\subseteq[n]$ and Bob has a subset $B\subseteq[n]$, and they must determine whether $A\cap B$ is empty or not by sending each other messages in any number of rounds. It is known that any randomized algorithm solving this task with probability at least $2/3$ requires $\Omega(n)$ bits of communication \cite{razborov1990distributional}. 

Suppose that there is a randomized turnstile streaming algorithm $\mathcal A$ which can output a relative error approximation to the $k$ means clustering cost with probability at least $2/3$ while using $r$ passes and space at most $M$. Then, we claim that Alice and Bob can use this algorithm to solve set disjointness in $2rM$ bits of communication, which implies that $M = \Omega(n/r)$. To do this, Alice first runs the algorithm $\mathcal A$ on the input stream which updates $\Ab_{i,1} \gets \Ab_{i,1} + 1$ for every $i\in A$. Then, Alice sends the memory state of $\mathcal A$, which is at most $M$ bits, to Bob. Bob then continues to run the algorithm $\mathcal A$ by updating running it on the stream which updates $\Ab_{i,1} \gets \Ab_{i,1} + 1$ for every $i\in B$. Finally, Bob also adds two dummy coordinates which has entries $0$ and $1$ each. Bob can then send the memory state back to Alice, which again is at most $M$ bits. This can be repeated for $r$ passes, for a total of $2rM$ bits of communication.

We now show that given an estimate $c$ satisfying \eqref{eq:cost-guarantee}, Alice and Bob can determine whether $A\cap B$ is empty or not. If $A\cap B$ is empty, then note that all rows of $\Ab$ are either $0$ or $1$, so the $k$-means clustering cost for $k = 2$ is $0$ and thus $c$ must be $0$. On the other hand, if $A\cap B$ is nonempty, then there is at least one row of $\Ab$ that is $2$ as well as a $0$ and a $1$ from the two dummy coordinates added by Bob, so the cost is strictly positive. Thus, $c$ must be strictly positive in this case. 
\end{proof}

\subsection{Proof of Theorem \ref{thm:lb-center}}
\label{sec:lb-center:proof}

\begin{proof}[Proof of Theorem \ref{thm:lb-center}]
Our proof for this result roughly follows our proof of Theorem \ref{thm:lb-cost}, so we only point out the important changes. We again let Alice and Bob have subsets $A\subseteq[n]$ and $B\subseteq[n]$, respectively. However, for this reduction, we construct our input instance $\Ab$ to be $(2n+3)\times 1$. First, Alice inserts her items $i\in A$ from $A$ in two coordinates, updating $\Ab_{2i,1}\gets \Ab_{2i,1} + 1$ and $\Ab_{2i+1,1}\gets \Ab_{2i+1,1} + 1$ for every $i\in A$. Similarly, Bob updates $\Ab$ in the two coordinates $\Ab_{2i,1}\gets \Ab_{2i,1} + 1$ and $\Ab_{2i+1,1}\gets \Ab_{2i+1,1} + 1$ for every $i\in B$. Finally, Bob inserts three dummy coordinates which has entries $0$, $1$, and $3$.

We now claim that an approximate set of centers $\tilde\Db$ can distinguish the cases between $A\cap B$ empty and $A\cap B$ nonempty. In the former case, the set of centers output by the $k$-means clustering algorithm must be $\{0,1,3\}$, since this is the unique solution with a cost of $0$. On the other hand, if $A\cap B$ is nonempty, then we claim that the $k$-means clustering algorithm cannot output $\{0,1,3\}$. Indeed, in this case, the cost of this solution is at least $2$ since there are at least two coordinates whose value is $2$. On the other hand, the solution of $\{0,1,2\}$ has a cost of $1$, since there is only a single dummy coordinate of $3$ that does not intersect exactly with these centers.
\end{proof}

\subsection{Proof of Theorem \ref{thm:random-order-bounded-sensitivity}}
\label{sec:random-order-bounded-sensitivity:proof}

We will need the following sensitivity sampling theorem:

\begin{theorem}[Sensitivity sampling, \cite{FL2011, BFL2016, WY2023}]
\label{thm:sensitivity-sampling-clustering}
Let
\[
    \tilde\sigma_i \geq \sup_{c^1, c^2, \dots, c^k\in\mathbb R^d}\frac{\min_{j=1}^k \norm{a^i - c^j}_2^2}{\sum_{i'=1}^n \min_{j=1}^k \norm{a^{i'} - c^j}_2^2}
\]
and $\tilde{\mathfrak S}\coloneqq \sum_{i=1}^n \tilde \sigma_i$. Suppose that for each $i\in[n]$, $a^i$ is sampled independently with probability $p_i \coloneqq \min\{1, \tilde O(\tilde\sigma_i kd/\epsilon^2)\}$, with an associated weight $w_i = 1/p_i$ if $i$ is sampled and $0$ otherwise. Then, for every $c^1, c^2, \dots, c^k\in\mathbb R^d$, we have that
\[
    \sum_{i=1}^n \min_{j=1}^k \norm{a^{i} - c^j}_2^2 = (1\pm\epsilon) \sum_{i=1}^n w_i \min_{j=1}^k \norm{a^{i} - c^j}_2^2.
\]
\end{theorem}

We then obtain the following result:

\begin{proof}[Proof of Theorem \ref{thm:random-order-bounded-sensitivity}]
Note that if a dataset has sensitivities bounded by $\alpha$, then a uniformly random sample of size $\tilde O(\alpha nkd/\epsilon^2)$ is a sample as given in Theorem \ref{thm:sensitivity-sampling-clustering}. Thus, approximately optimal centers $\hat c^1, \hat c^2, \dots, \hat c^k \in\mathbb R^d$ are approximately optimal centers for the entire dataset. These centers can be found using just
\[
    \tilde O((\alpha nkd/\epsilon^2)/\epsilon^2 + dk/\epsilon) = \tilde O(\alpha nkd/\epsilon^4 + dk/\epsilon)
\]
bits of space, using our turnstile streaming $k$ means clustering result (Theorem \ref{thm:k-means-turnstile-streaming}). Furthermore, because the input stream is a random order stream, these approximately optimal centers $\hat c^1, \hat c^2, \dots, \hat c^k$ can be obtained after seeing the first $\tilde O(\alpha nkd/\epsilon^2)$ elements of the stream. With approximately optimal centers in hand, note that the rest of the $n - \tilde O(\alpha nkd/\epsilon^2)$ points can be assigned on the fly, and thus space complexity is just an additional $O(n\log k)$ bits.
\end{proof}

\end{document}

%% file: neurips_upper.tex
\section{Fixed parameter PTAS for sparse dictionary learning}
\label{sec:ptas}

\subsection{PTAS for $r$-sparse dictionary learning}

In this section, we provide an algorithm which solves the $r$-sparse dictionary learning problem (Definition \ref{def:sparse_dict}) in time polynomial in the input matrix size $(n)$ and dimension $(d)$ up to $\epsilon$-relative error, for fixed $k$ and $\epsilon$. Additionally, we show that a similar approach can be used to provide an algorithm for $k$-means (Definition \ref{def:k-means}) that matches the current best dependency on $n,d,\epsilon$ and $k$ up to lower terms. First, we introduce a dimensionality reduction method that applies to both problems. 

\subsection{Dimensionality reduction}\label{sec:dim_reduction}

Our first step is to reduce the dimensionality of the given problem.  Since the only difference between $k$-means and sparse dictionary learning is the constraint on the left factor, $\Xb$, we can use the same sketching approach to reduce both problems.  Consider the following general definition:

\textbf{General problem}:
Let $\Xcal \subset \R^{n \times k}$ and $\Ab \in \R^{n \times d}$. Let $k \ll n,d$. Define the optimal solution as:
\begin{gather}\label{eq:general_frobenius}
    (\Xb^*, \Db^*)  = \argmin_{\Xb \in \Xcal, \Db \in \R^{k \times d}} \|\Xb\Db - \Ab\|_F^2
\end{gather}

The following theorem states that one may efficiently reduce the dimensionality of $\Ab$ in sparse dictionary learning or $k$-means. We briefly sketch the ideas behind the reduction. Intuitively, the regression guarantee of Theorem 3.1 in \cite{clarkson2009numerical} states that if $\Sb$ is a rank $k \ll d$ $\ell_2$-embedding matrix, then $\Dbtil = \argmin_{\Db \in \R^{k \times d}}\|\Sb(\Xb^*\Db - \Ab)\|_F^2$ will be a good approximation to the optimal solution of the original problem. While we do not know $\Xb^*$, this guarantee implies that there is an approximately optimal dictionary, $\Dbtil$, 
in the row %
space of $\Sb\Ab$. We can then restrict the optimization problem to consider only dictionaries in this lower dimensional space. Therefore, we only need to consider the error residual in this lower dimensional space, so we may reduce the dimension of the problem by applying an affine-embedding matrix $\Tb$ 
and then applying SVD to find the dominant singular subspace of $\Sb\Ab\Tb$. Finally, we project the rows of $\Ab$ to this dominant subspace. We can then solve the lower dimensional problem and map the solution to the original space.
\begin{theorem}\label{thm:dim_reduction}
    There is an algorithm which solves the problem in \eqref{eq:general_frobenius} up to $\epsilon \in (0,1)$ relative error with constant probability in $\Ocal(\nnz{\Ab} + (n + d)\poly(k/\epsilon))$ time plus the time needed to solve:
    \begin{gather*}
        \min_{\Xb \in \Xcal, \Db \in \R^{k \times s}} \|\Xb\Db - \Ab'\|_F^2,
    \end{gather*}
    to within $\epsilon$-relative error for $s = \Ocal(k\log(k)/\epsilon )$ and some $\Ab' \in \R^{n \times s}$ with constant probability.
\end{theorem}

In the rest of this section, we assume that $d = \poly(k/\epsilon)$ for clearer exposition, since the above theorem implies we can reduce to this case efficiently.

\subsection{Algorithm for sparse dictionary learning}

The first component of our algorithm for sparse dictionary learning is a coreset construction that reduces the size of the problem from $n$ to a size that is logarithmic in $n$. We achieve this by first leveraging an existing coreset construction for projective clustering by \cite{tukan2022new}. In the $(\ell, m)$-projective clustering problem, the goal is to find a set of $\ell$ $m$-dimensional subspaces that minimizes the sum of the squared Euclidean distances of the input vectors $\{a^i\}_{i=1}^n$ to the closest subspace. Observe that, in the $r$-sparse dictionary problem, the minimum cost of a dictionary is the sum of the squared Euclidean distances of the input vectors to the ${k \choose r}$ subspaces spanned by any subset of $r$ vectors of the $k$ vectors in the dictionary.  Hence, a coreset which preserves the projective clustering cost when $\ell = {k \choose r}$ will also preserve the cost of a dictionary in sparse dictionary learning.

After applying the coreset, we have reduced the size of the sparse dictionary problem to be at most logarithmic in $n$. This allows us to guess the sparsity pattern of the optimal left factor $\Xb^*$, since at most $r$ entries in each row of $\Xb^*$ may be nonzero.  For each guess of the sparsity pattern of $\Xb^*$, we can find an approximately optimal solution under this constraint by recognizing this as a polynomial optimization problem. We apply the decision algorithm of \cite{renegar1992computational_decision} using binary search to determine each entry of $\Db$ and the nonzero entries of $\Xb$ as done in \cite{RSW2016}. At some point we guess the sparsity pattern of $\Xb^*$, and hence attain an $\epsilon$-relative error solution to the sparse dictionary problem. The next theorem formally states the assumptions and guarantees of our algorithm, which is formalized in Algorithm \ref{alg:ptas_rsparse_dict} in the appendix.

\begin{theorem}\label{thm:sparse_dict_ptas}
    For an input for the $r$-sparse dictionary learning problem (Definition \ref{def:sparse_dict}) with error tolerance $\epsilon \in (0,1)$ such that the entries of $\Ab$ have bounded bit complexity, Algorithm \ref{alg:ptas_rsparse_dict} returns $\Xbtil \in \Xcal$ and $\Dbtil \in \R^{k \times d}$ satisfying:
    \begin{gather*}
        \|\Xbtil\Dbtil - \Ab\|_F \leq (1+\epsilon) \min_{\Xb\in\Xcal, \Db\in\R^{k \times d}} \|\Xb\Db - \Ab\|_F,
    \end{gather*}
    in $\poly(n)$ time with constant probability, when $k$, $r$, and $1/\epsilon$ are bounded by a constant.\footnote{If $k$ and $r$ are not assumed to be constant, then the time complexity is $\exp((8k^{3r})^{O(k^{2r+1})}\log n)$.}
\end{theorem}

\subsection{Algorithms for $k$-means}\label{sec:ptas_kmeans}
The same general approach of applying dimensionality reduction and a coreset construction along with guessing the sparsity pattern of $\Xb^*$ can be used to achieve a fixed-parameter PTAS for $k$-means as well.  However, we can achieve an improved time complexity matching the current best dependency on $k$ and $\epsilon$ up to lower order terms by further reducing the problem using results on leverage score sampling. Specifically, we combine Theorem 17 in \cite{woodruff2014sketching} and Theorem 3.1 in \cite{clarkson2009numerical} to prove the following lemma.
\begin{lemma}\label{lemma:leverage_regression}
    There is a set of matrices $\Scal \subset \R^{s \times n}$ with exactly one non-zero entry per column such that for any $\Ab \in \R^{n \times k}$ and $\Bb\in\R^{n \times d}$, there exists $\Sb \in \Scal$, so that if:
    \begin{gather*}
        \Xbtil = \argmin_{\Xb \in \R^{k \times d}} \|\Sb(\Ab\Xb - \Bb)\|_F
        ~~\text{and}~~
        \Xb^* = \argmin_{\Xb \in \R^{k \times d}} \|\Ab\Xb - \Bb\|_F,
    \end{gather*}
    then, 
    \begin{gather*}
        \|\Ab\Xbtil - \Bb\|_F \leq (1+\epsilon)\|\Ab\Xb^* - \Bb\|_F.
    \end{gather*}
    Furthermore, $\Scal$ depends only on $n$, $k$, and $\epsilon$; and $|\Scal| = n^{\Ocal(\frac{k \log k}{\epsilon})}$.
\end{lemma}
After applying a coreset construction to reduce the $k$-means problem to size $\poly(k/\epsilon)$, we can efficiently apply the above lemma to then reduce the problem to size $\Ocaltil(k/\epsilon)$.  Then, we brute force over all possible left-factors to find $\Xb^*$.  The following theorem states our results formally.
\begin{theorem}\label{thm:kmeans_result}
    For any input $\Ab \in \R^{n \times d}$ and $\epsilon \in (0, 1)$, Algorithm \ref{alg:ptas_kmeans} will return a feasible solution to the $k$-means clustering problem (Definition \ref{def:k-means}), $(\Xbtil, \Dbtil),$ satisfying:
    \begin{gather*}
        \|\Xbtil\Dbtil - \Ab\|_F \leq (1+\epsilon)\cdot \min_{\Db\in\R^{k\times d}, \Xb\in\Xcal} \|\Xb\Db - \Ab\|_F,
    \end{gather*}
    with constant probability. Furthermore, Algorithm \ref{alg:ptas_kmeans} runs in $n\cdot\poly(k/\epsilon) + \exp(\frac{k}{\epsilon} \polylog(k/\epsilon))$ time.
\end{theorem}

\section{Turnstile streaming algorithms}
\label{sec:turnstile-streaming}

In this section, we consider the the \emph{turnstile streaming model} (see Definition \ref{def:turnstile-stream}). We provide upper bounds on the space needed to compute an $\epsilon$-relative error solution to the $k$-means problem and a restricted form of the sparse dictionary learning problem in a turnstile stream. We do this by showing that these approximately optimal solutions can be computed from a few small linear sketches of the original data matrix, and any linear sketch can be trivially maintained in a turnstile stream by linearity of the updates. A key idea behind these algorithms is applying the \emph{guess-the-sketch} approach introduced in \cite{RSW2016} along with the following theorem.
\begin{theorem}\label{thm:regression_approx_dense}
    (Theorem 3.1 in \cite{clarkson2009numerical}) Given $\delta,\epsilon > 0$, suppose $\Ab$ and $\Bb$ are matrices with $n$ rows, and
    $\Ab$ has rank at most $k$. There is an $m = O(k \log(1/\delta)/\epsilon)$ such that, if $\Sb$ is an $m \times n$ sign matrix, then with probability at least $1 - \delta$, if $
        \Xbtil = \argmin_{\Xb}\|\Sb(\Ab\Xb - \Bb)\|_F^2$
        and
        $\Xb^* = \argmin_{\Xb}\|\Ab\Xb - \Bb\|_F^2,$ then
        $\|\Ab\Xbtil - \Bb\|_F \leq (1+\epsilon)\|\Ab\Xb^* - \Bb\|_F$.
\end{theorem}

Notice that, if we knew the optimal solution $\Xb^*$ exactly, then by the previous theorem we could compute an approximately optimal dictionary $\Dbtil$ exactly as $\Dbtil = (\Sb\Xb^*)^\dagger\Sb\Ab$. The key observation is that, since $\Sb$ is a random sign matrix and the rows of $\Xb$ are standard basis vectors, the set $\{\Sb\Xb ~|~\Xb\in\Xcal, \Sb\in\{\pm 1\}^{\Ocaltil(k/\epsilon) \times n}\}$ is not too large. Also, we can approximately solve $\min_{\Xb\in\Xcal}\|\Xb\Dbtil - \Ab\|_F^2$ for a fixed $\Dbtil$ with constant probability by solving $\Xbtil = \min_{\Xb\in\Xcal}\|(\Xb\Dbtil - \Ab)\Tb\|_F^2$, where $\Tb$ is a moderately sized affine embedding matrix. Since the number of possible $(\Xbtil, \Dbtil)$ is not too large, an $\ell_2$-embedding matrix, $\Wb$, can be used to approximate $\|\Xbtil\Dbtil - \Ab\|_F^2$ for every possible $(\Xbtil, \Dbtil)$. 

Our streaming algorithm relies on carefully balancing the roles of the three sketching matrices to minimize the size of the sketches, using the weakest guarantee possible for each component.  In particular, it is critical to use the affine embedding matrix $\Tb$ to only preserve the error for a fixed $\Dbtil$ instead of every subproblem and instead use the $\ell_2$-embedding matrix $\Wb$ to identify which subproblem provides an approximate solution to the overall problem.

\begin{theorem}
\label{thm:k-means-turnstile-streaming}
    (1) There are distributions of random sketching matrices $\Tb \in \R^{d \times t}$, $\Sb \in \R^{s \times n}$, and $\Wb \in \R^{w \times nd}$, with $t = \Ocal(\log(nk)/\epsilon^2)$, $s = \Ocal(\frac{k}{\epsilon})$, and $w = \Ocal(\frac{k^2}{\epsilon^3} \log(n))$ such that $\Sb\Ab$, $\Ab\Tb$, and $\Wb\operatorname{vec}(\Ab)$ suffice to compute a $(1+\epsilon)$-approximate solution to the $k$-means problem with at least constant probability, where $\operatorname{vec}(\Ab)\in\mathbb R^{nd}$ is the flattening of $\Ab$. 
    
    (2) There is an algorithm which computes a $(1+\epsilon)$-approximate solution to the $k$-means problem in the turnstile model with at least constant probability using $\Ocaltil(n/\epsilon^2 + dk/\epsilon)$ space for $n,d > \poly(k/\epsilon)$ in $n^{\Ocaltil(k^2/\epsilon)}$ additional time.
\end{theorem}

The previous proof critically relies on the fact that $\{\Sb\Xb ~|~\Xb\in\Xcal,~\Sb\in\{\pm 1\}^{m \times n}\}$ is a finite set that is not too large.  We must therefore introduce the following restricted form of the sparse dictionary problem.
\begin{definition}\label{def:discret_dict}
    (Discrete $r$-sparse dictionary problem) Let $\Xcal$ be the space of $n \times k$ matrices with at most $r$ non-zero entries per row and non-zero entries taking values in $\{-D, -(D-1),...,-1,0,1,...(D-1),D\}$. The goal of this problem is solve the following optimization problem:
    \begin{gather*}
        \Xb^*, \Db^* = \argmin_{\Xb \in \Xcal, \Db \in \R^{k \times d}} \|\Xb\Db - \Ab\|_F,
    \end{gather*}
    where $\Ab \in \R^{n \times d}$ is an arbitrary input matrix.
\end{definition}
Under this constraint that the solution is in a discrete space the proof of the streaming algorithm for sparse dictionary learning proceeds essentially the same as for $k$-means while accounting for the larger solution space.
\begin{theorem}\label{thm:sparse_dict_turnstile}
    (1) There are distributions of random sketching matrices $\Tb \in \R^{d \times t}$, $\Sb \in \R^{s \times n}$, and $\Wb \in \R^{w \times nd}$, with $t = \Ocal(r \log(nkD)/\epsilon^2)$, $s = \Ocal(\frac{k}{\epsilon})$, and $w = \Ocal(\frac{k^2}{\epsilon^3} \log(nD))$ such that $\Sb\Ab$, $\Ab\Tb$, and $\Wb\operatorname{vec}(\Ab)$ suffice to compute a $(1+\epsilon)$-approximate solution to the discrete $r$-sparse dictionary problem (Definition \ref{def:discret_dict}) with at least constant probability. 
    
    (2) There is an algorithm which computes a $(1+\epsilon)$-approximate solution to the $r$-sparse dictionary problem in the turnstile model with at least constant probability using $\Ocaltil(nr/\epsilon^2 + dk/\epsilon)$ space for $n,d > \poly(k/\epsilon)$ in $k^r \cdot (nD)^{\Ocaltil(k^2/\epsilon)}$ additional time.
\end{theorem}
Removing the restriction that $\Xb^*$ belongs to the restricted space would be an interesting future problem. However, two issues are that the entries of $\Xb$ may be very large, since the rows of $\Db$ may not be orthogonal, and a uniform discretization is required to apply a guess-the-sketch argument.

%% file: neurips_lower.tex
\section{Streaming lower bounds for Euclidean \texorpdfstring{$k$}{k}-means clustering}
\label{sec:lower-bounds}

We introduce slightly different definitions of the $k$-means clustering problem than the one used in Definition \ref{def:k-means} to facilitate the notation of our lower bound arguments in this section.

\begin{definition}[$k$-means clustering cost]
\label{def:k-means-cost}
Let $\{a^i\}_{i=1}^n\subseteq\mathbb R^d$ be a set of $n$ vectors in $d$ dimensions. Then, we define the $k$-means clustering cost of centers $c^1, c^2, \dots, c^k\in\mathbb R^d$ to be
\[
    \cost(c^1, c^2, \dots, c^k) \coloneqq \sum_{i=1}^n \min_{j=1}^k \norm{a^i - c^j}_2^2.
\]
\end{definition}

\begin{definition}[Approximate solutions to $k$-means clustering]
Let $\{a^i\}_{i=1}^n\subseteq\mathbb R^d$ be a set of $n$ vectors in $d$ dimensions. Let
\[
    \OPT \coloneqq \min_{c^1, c^2, \dots, c^k\in\mathbb R^d} \cost(c^1, c^2, \dots, c^k)
\]
We say that an algorithm outputs an $\epsilon$-approximate solution to the $k$-means clustering problem if the algorithm outputs one of the following:
\begin{itemize}
\item \textbf{Partition}: a partition $C^1, C^2, \dots, C^k\subseteq[n]$ such that
\[
    \sum_{j=1}^k \sum_{i\in C^j} \norm{a^i - \hat c^j}_2^2 \leq (1+\epsilon)\OPT
\]
where $\hat c^j \coloneqq \frac1{\abs{C^j}} \sum_{i\in C^j} a^i$. 
\item \textbf{Centers}: centers $\hat c^1, \hat c^2, \dots, \hat c^k\in\mathbb R^d$ such that $\cost(\hat c^1, \hat c^2, \dots, \hat c^k) \leq (1+\epsilon)\OPT$.
\item \textbf{Cost}: a number $c\geq 0$ such that $\OPT \leq c \leq (1+\epsilon)\OPT$.
\end{itemize}
\end{definition}

\subsection{Lower bounds for \texorpdfstring{$k$}{k}-means clustering}

Our most technically involved and delicate lower bound result is the following theorem, which shows that nearly optimally solving $k$-means clustering to $(1+\epsilon)$ accuracy requires $\tilde\Omega(n/\epsilon)$ bits of space:

\KMeansLB*

We defer the full proof to Appendix \ref{sec:n/eps-full-proof} and give a proof sketch in this section to illustrate the most important ideas. 

\paragraph{The hard instance: set disjointness.} The starting point to our lower bound is the information theoretic communication complexity lower bound for the set disjointness problem due to \cite{BJKS2004}. In the two-party set disjointness problem, two players Alice and Bob each have a bit vector $A, B\in\{0,1\}^d$ in $d$ dimensions, and they must determine whether there exists a coordinate $j\in[d]$ such that $A_j = B_j = 1$ or not. The work of \cite{BJKS2004} shows that in order to solve this problem, Alice and Bob must exchange messages that reveal at least $\Omega(d)$ bits of information about their inputs, which in turn implies an $\Omega(d)$ communication complexity lower bound for this problem, as well as an $\Omega(nd)$ communication complexity lower bound for solving a constant fraction of $n$ independent instances of the same problem. Furthermore, the hard instance of \cite{BJKS2004} has a simple input distribution: the vectors $(A, B)$ are such that the $j$th coordinate $(A^j, B^j)$ is drawn either as $(0, 0)$ with probability $1/2$ or $(1, 0)$ with probability $1/4$ or $(0, 1)$ with probability $1/4$, except for one coordinate, which may take the value $(1, 1)$. 

We aim to make use of this result as follows. Consider the vector $Z = A + B$. This vector has entries in $\{0,1\}$, except possibly for one entry, which could be $2$. If we have $n$ such vectors, then we expect a good clustering into $k = d$ clusters to cluster all points with $Z_j = 2$ together. Such a clustering would be able to output the \emph{index} of the intersection of $A$ and $B$, which intuitively requires more information than just determining whether there is an intersection or not, and thus should also require $\Omega(d)$ bits of information cost. Furthermore, we can choose the dimension $d$ to be roughly $1/\epsilon$, so that the cost of clustering $Z$ to the ``correct'' center will have a cost of $\Theta(d) = \Theta(1/\epsilon)$, while clustering $Z$ to the incorrect center will incur an additional error of $\Theta(1)$, which is an $\epsilon$ fraction of the cost.

\paragraph{Cost calculations.} The main challenge in carrying out the idea in the previous paragraph is in arguing that the target optimal clustering that we wish to discover indeed is a nearly optimal clustering, and that significant deviations from this clustering result in a large cost. This involves showing a lower bound on the cost of \emph{any} clustering. 

Our first step is to obtain a lower bound on the cost of any clustering of $n$ random bit vectors in $d$ dimensions. If we first fix a set of $k$ centers $\{c^j\}_{j=1}^k$, then the minimum distance between a random bit vector $Z$ and any of the $c^j$ can be bounded by using Chernoff bounds, which implies a lower bound of $d/4 - O(\log d)$ on this quantity in expectation (Lemma \ref{lem:expectation-bound}). Note, however, that this lower bound is not high enough to prevent a nearly optimal solution from just assigning points according to the best clustering of the random bits while ignoring the one entry that takes the value of $Z_j = 2$, which means that the clustering need not solve the problem of identifying the intersection coordinate between $A$ and $B$.

To address this problem, we need to make the cost of ignoring the intersection coordinate much more costly. We do this by instead considering the \emph{multi-party} set disjointness problem, so that we now have $t = O(\sqrt{\log d})$ players rather than just $2$, each with an input vector $A^{(i)}\in\{0,1\}^d$, so that $Z = \sum_{i=1}^t A^{(i)}$ is now a random bit vector except for a single entry with a $t$ rather than a $2$. Now, a clustering which does not correctly identify the intersection coordinate will pay a cost of roughly $t^2 = O(\log d)$, which is large enough to overcome the potential savings from a good clustering of the random bit coordinates. We also ``plant'' the target centers $c^j$ by adding roughly $n/k$ copies of each of our target centers $c^j$ as part of the input instance (Lemma \ref{lem:planted-centers}), so that choosing centers $\hat c^j$ that are significantly different from $c^j$ must incur a large cost. In particular, we can get the guarantee that on average, $\norm{c^j - \hat c^j}_2^2 \leq o(1)$.

At this point, we can argue that most of the $k$ centers are the centers that expect, i.e., roughly $t$ on one coordinate and $1/2$ on the rest of the coordinates. Thus, if we cluster a point $Z$ whose center we expect to be $\hat c^j$ but is clustered to some other $\hat c^{j'}$, and furthermore $\hat c^{j'}$ is close to our expected center $c^{j'}$, then we must incur an additional $O(\log d)$ cost which is too expensive. However, there is still the possibility that for the very small number of clusters $\hat c^j$ which do not satisfy $\norm{c^j - \hat c^j}_2^2 \leq o(1)$, these centers could be assigned a very large number of points with very low cost. We also show that this cannot be the case, by arguing that if a large number of points are assigned to very few clusters, then the cost must be large (Lemma \ref{lem:few-points-per-center}). With this lemma in hand, we are able to show our main result in Theorem \ref{thm:k-means-lb} by carefully combining the various cost contribution bounds discussed previously.

\subsubsection{Lower bound for outputting nearly optimal centers}

We note that an $\Omega(dk/\epsilon)$ lower bound follows from an earlier lower bound for low rank approximation due to \cite{Woo2014}, even for row arrival streams:

\begin{definition}[Row arrival stream]
\label{def:row-arrival-stream}
We say that an algorithm outputs an $\epsilon$-approximate solution to the $k$-means clustering problem in the row arrival streaming model if the input vectors $\{a^i\}_{i=1}^n\subseteq\mathbb R^d$ arrive one at a time.
\end{definition}

\begin{theorem}
\label{thm:cluster-center-lb}
Suppose that an algorithm outputs centers $\{\hat c^j\}_{j=1}^k\subseteq\mathbb R^d$ that achieves a $(1+\epsilon)$-approximately optimal solution to the $k$-means clustering problem after one pass through a row arrival stream (Definition \ref{def:row-arrival-stream}). Then, the algorithm must use at least $\tilde\Omega(dk/\epsilon)$ bits of space. 
\end{theorem}

We briefly justify why the techniques of \cite{Woo2014} imply Theorem \ref{thm:cluster-center-lb}. The result of \cite{Woo2014} constructs a distribution over $O(k/\epsilon)\times d$ matrices such that one can recover an arbitrary random bit among $\tilde\Omega(dk/\epsilon)$ random bits by appending a set of $k$ ``query'' rows and then computing a $(1+\epsilon)$-approximately optimal low rank approximation to the resulting matrix. Furthermore, it is shown that a nearly optimal rank $k$ approximation is obtained by approximating all but $k$ rows by zero vectors. Such a rank $k$ approximation in fact corresponds to a clustering solution, and thus the proof of \cite{Woo2014} immediately applies to our $k$-means clustering setting as well.

\subsection{Lower bounds for center cost query data structures}

Next, we study lower bounds against streaming algorithms which have the guarantee of approximating the cost of an arbitrary but fixed set of centers. We formalize the guarantee we study in Definition \ref{def:cost-query-data-structure}.

\begin{definition}[Center cost query data structure]
\label{def:cost-query-data-structure}
We say that $\mathcal Q$ is an \emph{$\epsilon$-approximate center cost query data structure} for the $k$ means clustering problem for the instance $\{a^i\}_{i=1}^n$ if, for any centers $c^1, c^2, \dots, c^k\in\mathbb R^d$, $\mathcal Q$ outputs one of the following:
\begin{itemize}
\item \textbf{Partition}: a partition $C^1, C^2, \dots, C^k\subseteq[n]$ such that
\[
    \sum_{j=1}^k \sum_{i\in C^j} \norm{a^i - c^j}_2^2 \leq (1+\epsilon)\cost(c^1, c^2, \dots, c^k).
\]
\item \textbf{Cost}: a number $c\geq 0$ such that
\[
    \cost(c^1, c^2, \dots, c^k) \leq c \leq (1+\epsilon)\cost(c^1, c^2, \dots, c^k)
\]
\end{itemize}
\end{definition}

Our first lower bound is an $\Omega(n/\epsilon^2)$ bit space lower bound for a center cost query data structure which can output a partition for $k$-means clustering with $k=2$. We proceed by a standard encoding argument, showing that any such data structure must encode $\Omega(n/\epsilon^2)$ many random bits. We provide the full proof in Appendix \ref{sec:strong-lb-proof}.

\begin{theorem}
\label{thm:strong-lb}
Let $\epsilon\in(0, 1/3)$ and $k = 2$. Suppose that an algorithm maintains an $\epsilon/15$-approximate center cost query data structure for $k$-means clustering that outputs a partition (Definition \ref{def:cost-query-data-structure}) over a row arrival stream (Definition \ref{def:row-arrival-stream}). Then, the algorithm must use at least  $\Omega(n/\epsilon^2)$ bits of space, over any constant number of passes.
\end{theorem}

\subsection{Approximation of costs and centers}

We show $\Omega(n)$ space memory bounds when we only need to estimate the optimal cost or centers achieving nearly optimal cost, up to a constant factor. Our lower bounds in this section are simpler reductions from the set disjointness problem \cite{razborov1990distributional, BJKS2004}. Proofs are provided in Appendix \ref{sec:lb-cost:proof} and \ref{sec:lb-center:proof}.

\begin{theorem}[Lower Bound for Estimating $k$-Means Clustering Cost]
\label{thm:lb-cost}
Let $k = 2$ and let $\Xcal$ be the set of matrices $\Xb\in\mathbb R^{n\times k}$ with standard basis vectors as rows. Let $d = 1$. Any randomized algorithm which outputs a number $c\geq 0$ satisfying
\begin{equation}\label{eq:cost-guarantee}
    c\leq \min_{\Xb \in \Xcal, \Db \in \R^{k \times d}} \|\Xb\Db - \Ab\|_F^2 < 2c
\end{equation}
in a constant number of passes over a turnstile stream requires $\Omega(n)$ bits of space. 
\end{theorem}

\begin{theorem}[Lower Bound for Computing Approximate Centers]
\label{thm:lb-center}
Let $k = 3$ and let $\Xcal$ be the set of matrices $\Xb\in\mathbb R^{n\times k}$ with standard basis vectors as rows. Let $d = 1$. Any randomized algorithm which outputs centers $\tilde\Db\in\mathbb R^{k\times d}$ satisfying
\[
    \min_{\Xb \in \Xcal} \|\Xb\tilde\Db - \Ab\|_F^2 < 2\min_{\Xb \in \Xcal, \Db \in \R^{k \times d}} \|\Xb\Db - \Ab\|_F^2
\]
in a constant number passes over a turnstile stream requires $\Omega(n)$ bits of space.
\end{theorem}

\subsection{New upper bounds in random order streams}

In this section, we show some new upper bounds showing that we can go beyond the previously presented lower bounds. In particular, in random order row arrival streams with bounded sensitivity, we show that the first segment of the stream is sufficient to obtain approximately optimal centers, and these can in turn be used to nearly optimally cluster the rest of the stream. We give the full proof of this result in Appendix \ref{sec:random-order-bounded-sensitivity:proof}.

\begin{theorem}
\label{thm:random-order-bounded-sensitivity}
Suppose that the rows of $\Ab\in\mathbb R^{n\times d}$ arrive in a random order row arrival stream. Furthermore, suppose that the sensitivities of each row $a^i$ are bounded by $\alpha$, that is,
\[
    \sup_{c^1, c^2, \dots, c^k\in\mathbb R^d}\frac{\min_{j=1}^k \norm{a^i - c^j}_2^2}{\sum_{i'=1}^n \min_{j=1}^k \norm{a^{i'} - c^j}_2^2} \leq \alpha.
\]
Then, there is an algorithm which, with constant probability, outputs a $(1+\epsilon)$-nearly optimal clustering with partitions and centers using
\[
    \tilde O(\alpha nkd/\epsilon^4 + dk/\epsilon + n).
\]
bits of space. In particular, if $\alpha \leq \epsilon^4/kd$, then this algorithm uses just $\tilde O(n + dk/\epsilon)$ bits of space.
\end{theorem}

%% file: appendix.tex
\section{Missing proofs for Section \ref{sec:ptas}}

In this section, we provide the missing proofs for Theorem \ref{thm:dim_reduction}, Theorem \ref{thm:sparse_dict_ptas}, and Theorem \ref{thm:kmeans_result}, along with prerequisite definitions and results. We also provide Algorithm \ref{alg:ptas_rsparse_dict} and Algorithm \ref{alg:ptas_kmeans}.

Recall that, after introducing the dimensionality reduction result of Theorem \ref{thm:dim_reduction}, we assume $d = \poly(k/\epsilon)$ in subsequent sections for clearer exposition.

\subsection{Dimensionality reduction}

We first restate an affine embedding guarantee provided for the CountSketch matrix by prior work.
\begin{lemma}\label{lemma:affine_embedding}
    (From Lemma A.2 of \cite{liu2020learned}) Given matrices $\Ab$, $\Bb$ with $n$ rows, a sparse embedding matrix $\Sb$ (i.e., CountSketch) with $\Ocal(\operatorname{rank}(\Ab)^2/\epsilon^2)$ rows satisfies for all $\Xb$ of appropriate dimension with constant probability:
    \begin{gather*}
        \|\Sb(\Ab\Xb - \Bb)\| = (1\pm \epsilon)\|\Ab\Xb - \Bb\|_F^2
    \end{gather*}
     Moreover, the matrix product $\Sb \cdot \Ab$ can be computed in $\Ocal(\nnz{\Ab})$ time.
\end{lemma}

Next, we combine a few prior results to provide a regression error guarantee with a sketch that can be efficiently applied.
\begin{lemma}\label{thm:regression_approx}
    Given $\delta,\epsilon > 0$, suppose $\Ab$ and $\Bb$ are matrices with $n$ rows, and
    $\Ab$ has rank at most $k$. There is an $s = O(k \log(k)/\epsilon)$ and a random matrix $\Sb \in \R^{s \times n}$ such that, with high constant probability, if:
    \begin{gather*}
        \Xbtil = \argmin_{\Xb}\|\Sb(\Ab\Xb - \Bb)\|_F^2
        \quad\text{and}\quad
        \Xb^* = \argmin_{\Xb}\|\Ab\Xb - \Bb\|_F^2,
    \end{gather*}
    then,
    \begin{gather*}
        \|\Ab\Xbtil - \Bb\|_F \leq (1+\epsilon)\|\Ab\Xb^* - \Bb\|_F.
    \end{gather*}
    Furthermore, $\Sb\cdot\Ab$ can be computed in $\Ocal(\nnz{\Ab} + d\cdot\poly(k/\epsilon))$ time.
\end{lemma}
\begin{proof}
    We will define $\Sb \in \R^{s \times n}$ as $\Sb = \Gb\cdot \Cb$, where $\Gb \in \R^{s \times c}$ is a Gaussian sketching matrix and $\Cb \in \R^{c \times n}$ is a CountSketch matrix, where $c = \poly(k/\epsilon)$.  Note that $\Sb\Ab$ can be computed by first computing $\Cb\Ab$ in $\Ocal(\nnz{\Ab})$ time and then computing $\Gb\cdot\Cb\Ab$ in $\Ocal(d\cdot\poly(k/\epsilon))$ time.

    Our first step is to show that the distribution of $\Sb$ is an $\ell_2$-subspace embedding (see Definition 2 of \cite{woodruff2014sketching}). By Theorem 9 of \cite{woodruff2014sketching}, the distribution of $\Cb$ is an $\ell_2$-subspace embedding and by Theorem 6 of \cite{woodruff2014sketching}, the distribution of $\Gb$ is an $\ell_2$-subspace embedding, each with high constant probability.

    We can compose the $\ell_2$-subspace embedding guarantees to get the following bound with high probability via the union bound.
    \begin{gather*}
        (1-\epsilon)\|\xb\|_2 \leq \|\Cb\xb\|_2 \leq (1+\epsilon)\|\xb\|_2\\
        \Rightarrow (1-\epsilon)^2\|\xb\|_2 \leq \|\Gb\Cb\xb\|_2 \leq (1+\epsilon)^2\|\xb\|_2
    \end{gather*}
    Hence, $\Sb$ is an $\epsilon$-subspace embedding for a fixed $k$-dimensional space with high constant probability after adjusting $\epsilon$ by a constant factor. Therefore, $\|\Ub^T\Sb\Sb^T\Ub - \Ib\|_2 \leq \epsilon,$ with high constant probability. The rest of the proof is the same as the proof of Theorem 3.1 in \cite{clarkson2009numerical} while using this $\ell_2$-embedding matrix $\Sb$ instead of a random sign matrix.
\end{proof}
\textbf{Proof of Theorem \ref{thm:dim_reduction}}
\begin{proof}
    By Lemma \ref{thm:regression_approx}, there exists a random matrix $\Sb \in \R^{s \times n}$ for $s = \Ocal(\frac{k}{\epsilon}\log(k))$, such that, with at least constant probability,
    \begin{gather*}
        \Dbtil = \argmin_{\Db \in \R^{k \times d}}\|\Sb(\Xb^*\Db - \Ab)\|_F^2
        \Rightarrow
        \|\Xb^*\Dbtil - \Ab\|_F \leq (1+\epsilon)\|\Xb^*\Db^* - \Ab\|_F.
    \end{gather*}
    In this case, we can solve for $\Dbtil$ exactly as $\Dbtil = (\Sb\Xb^*)^\dagger\Sb\Ab$, hence, $\Dbtil = \Rb\Sb\Ab$ for some $\Rb \in \R^{k \times s}$.  Therefore, $\Dbtil = \Rbtil\Sb\Ab$, where,
    \begin{gather*}
        \Rbtil = \argmin_{\Rb\in\R^{k \times s}} \|\Xb^*\Rb\Sb\Ab - \Ab\|_F^2.
    \end{gather*}
    Let $\Tb_1 \in \R^{d \times \Ocal(s^2/\epsilon^2)}$ be a count sketch matrix. Since $\rank(\Sb\Ab) \leq s$, Lemma \ref{lemma:affine_embedding} guarantees that $\|\Mb\Sb\Ab\Tb_1 - \Ab\Tb_1\|_F^2 = (1\pm\epsilon)\|\Mb\Sb\Ab - \Ab\|_F^2$ for all $\Mb \in \R^{n \times s}$ simultaneously with at least constant probability. Since this holds for all $\Mb \in \R^{n \times s}$, and $\{\Xb\Db ~|~ \Xb \in \Xcal, \Db \in \R^{k \times s}\} \subset \R^{n \times s}$, we have that:
    \begin{gather}
        \Xbtil', \Rbtil' = \argmin_{\Xb \in \Xcal, \Rb\in\R^{k \times s}}  \|\Xb\Rb\Sb\Ab\Tb_1 - \Ab\Tb_1\|_F^2 \nonumber \\
        \Rightarrow \|\Xbtil'\Rbtil'\Sb\Ab - \Ab\|_F^2 \leq (1+\epsilon)\|\Xb^*\Rbtil\Sb\Ab - \Ab\|_F^2 = (1+\epsilon)\|\Xb^*\Dbtil - \Ab\|_F^2 \leq (1 + \epsilon)^2\|\Xb^*\Db^* - \Ab\|_F^2.\label{eqn:affine_sketching_step}
    \end{gather}
    However, note that $\Sb\Ab\Tb_1$ has rank of at most $s$.  Let $\Tb_2 \in \R^{\Ocal(s^2/\epsilon^2) \times s}$ be the top $s$ right singular vectors of $\Sb\Ab\Tb_1$, and let $\Tb = \Tb_1\Tb_2$, then,
    \begin{align*}
        \Xbtil', \Rbtil' &= \argmin_{\Xb \in \Xcal, \Rb\in\R^{k \times s}} \|\Xb\Rb\Sb\Ab\Tb_1 - \Ab\Tb_1\|_F^2 \\
        &= \argmin_{\Xb \in \Xcal, \Rb\in\R^{k \times s}} \|(\Xb\Rb\Sb\Ab\Tb_1 - \Ab\Tb_1)\Tb_2\Tb_2^T\|_F^2 + \|\Ab\Tb_1(\Ib - \Tb_2\Tb_2^T)\|_F^2 \\
        &= \argmin_{\Xb \in \Xcal, \Rb\in\R^{k \times s}} \|\Xb\Rb\Sb\Ab\Tb_1\Tb_2 - \Ab\Tb_1\Tb_2\|_F^2 \\
        &= \argmin_{\Xb \in \Xcal, \Rb\in\R^{k \times s}} \|\Xb\Rb\Sb\Ab\Tb - \Ab\Tb\|_F^2.
    \end{align*}
    Notice that $\{\Rb\Sb\Ab\Tb ~|~\Rb\in\R^{k \times s}\} = \R^{k \times s}$ with probability one if $\rank(\Ab) > s$. If it does not hold that $\rank(\Ab) > s$, then we may directly reduce the dimension of the problem by SVD. Therefore, we can instead solve:
    \begin{align*}
        \Xbtil', \Dbtil' &= \argmin_{\Xb \in \Xcal, \Db\in\R^{k \times s}} \|\Xb\Db - \Ab\Tb\|_F^2.
    \end{align*}
    By the above equations, $\Rbtil' = \Dbtil'(\Sb\Ab\Tb)^\dagger$ and by eqn. (\ref{eqn:affine_sketching_step}), $\|\Xbtil'\Rbtil'\Sb\Ab - \Ab\|_F^2 \leq (1 + \epsilon)^2\|\Xb^*\Db^* - \Ab\|_F^2$.  Therefore, we can return $\Xb = \Xbtil'$ and $\Db =\Db' (\Sb\Ab\Tb)^\dagger\Sb\Ab$ to guarantee:
    \begin{gather*}
        \|\Xb\Db - \Ab\|_F^2 \leq (1+\epsilon)^2\|\Xb^*\Db^* - \Ab\|_F^2 \leq (1+3\epsilon)\|\Xb^*\Db^* - \Ab\|_F^2.
    \end{gather*}
    Now we work out the time complexity of the above reduction.  First, we must compute $\Ab\Tb$ to reduce to the smaller optimization problem.  To do this, we can sample the CountSketch matrix $\Tb_1 \in \R^{k \times \Ocal(s^2/\epsilon^4)}$ and compute $\Ab\Tb_1$ in $\Ocal(\nnz{\Ab} + \poly(k/\epsilon))$ time.  Then, we sample the sketching matrix $\Sb\in \R^{\Ocal(k/\epsilon \cdot \log k) \times n}$ and compute $\Sb\Ab\Tb_1$ in $\Ocal(\nnz{\Ab} + \poly(k/\epsilon))$ time.  Then, we compute $\Tb_2$ via the SVD of $\Sb\Ab\Tb_1$ and compute $\Ab\Tb = \Ab\Tb_1\Tb_2$ in $\poly(k/\epsilon)$ time.  From here, we then solve the optimization problem for $\Xbtil'$ and $\Dbtil'$.

    To convert $\Dbtil'$ to an approximate solution to the original problem, we must compute $\Db = \Db'(\Sb\Ab\Tb)^\dagger\Sb\Ab$. We can compute $(\Sb\Ab\Tb)^\dagger$ via the SVD and then form $\Db'(\Sb\Ab\Tb)^\dagger$ in $\poly(k/\epsilon)$ time.  Then, we compute the matrix product $\Sb\Ab$ in $\Ocal(\nnz{\Ab})$ time. Finally, the matrix product $\Db'(\Sb\Ab\Tb)^\dagger\Sb\Ab$ can be computed in $\Ocal(d\cdot\poly(k/\epsilon))$ time.

    Therefore, the total time complexity of the reduction procedure is $\Ocal(\nnz{\Ab} + (n + d)\poly(k/\epsilon))$.
\end{proof}

\subsection{PTAS for sparse-dictionary}

\subsubsection{Coreset construction for Sparse Dictionary Learning} We begin by providing a coreset construction for the $r$-sparse dictionary learning problem, which we derive from coreset construction for the projective clustering problem defined here.

\begin{definition}\label{def:projective_clustering}
    ($(\ell, m)$-Projective clustering problem) Let $\Ab \in \R^{n \times d}$ be a matrix containing $n$ points. For a fixed sequence $\Fcal =\{F_1,...,F_\ell\}$, of $m$-dimensional subspaces, define:
    \begin{gather*}
        \cost(\Fcal, \Ab) = \sum_{i=1}^n \min_{F \in \Fcal} \operatorname{dist}(\Fcal, \Ab_i)^2,
    \end{gather*}
    where $\operatorname{dist}(\Ab_i, F)^2$ denotes the squared Euclidean distance of the $i$-th row of $\Ab$ to the fixed subspace $F$.
    
    The goal of the $(\ell, m)$-Projective clustering problem is to find a size $\ell$ collection of $m$-dimensional linear subspaces, $\Fcal^*$, that minimizes the above cost function, i.e., $\Fcal^* = \argmin_{\Fcal} \cost(\Fcal, \Ab)$. 
\end{definition}

We will use this to construct a reweighted form of the $r$-sparse dictionary problem with smaller size which we define next.
\begin{definition}
    (Weighted $r$-SDL) Let $\Xcal_r \subset \R^{n \times k}$ denote the set of matrices with at most $r$ non-zero entries per row. For a given input matrix $\Ab \in \R^{n \times d}$, such that $k \ll n,d$, and diagonal matrix $\Wb\in\R^{n \times n}$ return:
    \begin{gather}
        (\Xb^*, \Db^*)  = \argmin_{\Xb \in \Xcal, \Db \in \R^{k \times d}} \|\Wb(\Xb\Db - \Ab)\|_F^2.
    \end{gather}
    The parameter $k$ is the number of dictionary elements and the parameter $r$ determines how many dictionary elements can be used to represent each row of $\Ab$.
\end{definition}

\begin{theorem}\label{thm:dictionary_coreset}
    Let $r$, $k$, $\Ab$, and $\Xcal$ be defined as in the sparse dictionary learning problem (Definition \ref{def:sparse_dict}). If the entries of $\Ab$ can each be represented by $b$ bits, then there exists an algorithm which computes a diagonal matrix $\Wb \in \R^{w \times w}$ and $\Ab' \in \R^{w \times d}$ in $\Ocal(n^2 k^{4r} b^{k^{2r + 1}})$ time, such that,
    \begin{gather*}
        \Big| \min_{\Xb\in\Xcal} \|\Wb(\Xb\Db - \Ab')\|_F^2 - \min_{\Xb\in\Xcal} \|\Xb\Db - \Ab\|_F^2 \Big| \leq \epsilon\cdot\min_{\Xb\in\Xcal} \|\Xb\Db - \Ab\|_F^2,
    \end{gather*}
    for all $\Db \in \R^{k \times d}$.  Furthermore, $w = \Ocal((8k^{3r}b\log d)^{\Ocal(k^{r+1})}\log n)$.
\end{theorem}
\begin{proof}
    First, we observe that any coreset for the $(\ell, m)$-projective clustering problem (Definition \ref{def:projective_clustering}) with $\ell = {k \choose r}$ and $m = r$ provides a coreset for the $r$-sparse dictionary learning problem. This is because if the collection of subspaces $\Fcal$ contains all $r$-dimensional subspaces spanned by $r$ rows of the dictionary $\Db$, then $\min_{\Xb \in \Xcal}\|\Xb\Db - \Ab\|_F^2 = \cost(\Fcal, \Ab)$.

    By Theorem 1.2\footnote{We have confirmed through correspondence to the authors that there is a typo in Definition 1.9 of \cite{tukan2022new}, and the definition should also state $(1-\epsilon)\sum_{\pb\in C}w(\pb)\dist(H(\Xb,\vb),\pb)^2 \leq \sum_{\pb\in C} \dist(H(\Xb,\vb),\pb)^2$. That is, Definition 1.9 defines a standard relative error coreset guarantee in the $\ell_2^2$-norm.} and Theorem 3.3 in \cite{tukan2022new}, Algorithm 2 of \cite{tukan2022new} outputs a set of points $\Pcal$ and weight function $w(p):\Pcal \rightarrow \R$ such that:
    \begin{gather*}
        \Big|\cost(\Fcal, \Ab) - \sum_{p \in \Pcal} w(p)\cdot\min_{F \in \Fcal} \operatorname{dist}(\Fcal, \Ab_i')^2\Big| \leq \epsilon \cdot \cost(\Fcal, \Ab),
    \end{gather*}
    for all $\Fcal$ that are a $j$-size sequence of $k$-dimensional subspaces.

    If $\Fcal$ is the collection of all $r$-dimensional subspaces spanned by $r$ rows of the dictionary $\Db$, then we can rewrite the above guarantee in matrix notation as follows:
    \begin{gather*}
        \Big| \min_{\Xb\in\Xcal} \|\Wb(\Xb\Db - \Ab')\|_F^2 - \min_{\Xb\in\Xcal} \|\Xb\Db - \Ab\|_F^2 \Big| \leq \epsilon\cdot\min_{\Xb\in\Xcal} \|\Xb\Db - \Ab\|_F^2,
    \end{gather*}
    where $\Ab_i'$ is the $i$-th point in the point set $\Pcal$ and $\Wb \in \R^{w \times w}$ is a diagonal matrix where $\Wb_{ii}$ is the weight $w(p_i)$.

    Theorem 1.2 of \cite{tukan2022new} then guarantees that $w = \Ocal((8\ell^3\log(d\Delta))^{\Ocal(\ell m)}\log n)$, where $\Delta$ is the the ratio of the largest and smallest non-zero entry magnitudes of $\Ab$. Therefore, $\Delta \leq 2^b$, and so $w = \Ocal((8\ell^3 b\log d)^{\Ocal(\ell m)}\log n)$. Furthermore, by the discussion below Theorem 3.3 of \cite{tukan2022new}, their algorithm runs in $\Ocal(n^2 \ell^4 (\log\Delta)^{\ell^2 m}) = \Ocal(n^2 \ell^4 b^{\ell^2 m})$ time.  Substituting in $\ell = k^r \geq {k \choose r}$ and $m = r$ to these bounds gives the final theorem statement.
    
\end{proof}

\subsubsection{Polynomial Solver for a Restricted SDL Problem}

Next, we show that by adding a further restriction on the weighted $r$-SDL problem, we can solve the problem in polynomial time. First, define the \emph{sparsity pattern} $\Ncal \in \{(\Ncal_i)_{i\in[n]} ~|~ |\Ncal_i| = r, \Ncal_i \subset [k]\}$, and let $\Xcal_\Ncal$ to be the set of $n \times k$ matrices such that $\Xb_{ij} = 0$ if $j \not\in \Ncal_{ij}$ for all $\Xb \in \Xcal_\Ncal$.  That is, $\Xb \in \Xcal_\Ncal$ is a matrix where only $r$ fixed entries per row may be non-zero, and these entries are specified by the sparsity pattern $\Ncal$. We define the following restricted solver.
\begin{definition}\label{def:poly_solver}
    For a given $r$-SDL problem, let $\texttt{PolySolver}$ be an algorithm which takes as input a sparsity pattern $\Ncal$, diagonal matrix $\Wb \in \R^{n \times n}$, input matrix $\Ab \in \R^{n \times d}$, dictionary size $k$, sparsity $r$, and error tolerance $\epsilon\in(0,1)$. If $\Ncal$ is the sparsity pattern of the optimal left-factor $\Xb^*$, then $\texttt{PolySolver}$ outputs $\Xbtil \in \Xcal_\Ncal$ and $\Dbtil \in \R^{k \times d}$ which satisfy: 
    \begin{gather*}
        \|\Wb(\Xbtil\Dbtil - \Ab)\|_F^2 \leq (1+\epsilon) \cdot  \|\Xb^*\Db^* - \Ab\|_F^2.
    \end{gather*}
\end{definition}

\begin{lemma}\label{lemma:polynomial_solve_dict}
    There exists an implementation of \texttt{PolySolver} that runs in $\Ocal(2^{\Ocal(nr+kd)})$ time given that the entries of $\Ab$ have bounded bit complexity.
\end{lemma}
\begin{proof}
    For $i \in [n]$ and $j\in[r]$, let $x_{ij}$ denote the $j$-th smallest entry in $\Ncal_i$ of a matrix $\Xb \in \Xcal_\Ncal$. Observe that the entry $[\Xb\Db]_{st}$ has the form $\sum_{j=1}^r x_{sj}\Db_{\Ncal_{s,j},t}$ , hence $[\Wb(\Xb\Db - \Ab)]_{st} = \Wb_{ss}(\sum_{j=1}^r x_{sj}\Db_{\Ncal_{s,j},t} - \Ab_{st})$. Therefore, $\|\Wb(\Xb\Db - \Ab)\|_F^2$ is a fourth degree polynomial in the set of variables $\{x_{ij}~|~i \in [n],~j\in[r]\}$ and the entries of $\Db$.

    By \cite{renegar1992computational_decision}, for a given polynomial $P(y_1,y_2,...,y_v)$ of degree $t$, we can determine whether there exists a solution satisfying $P(y_1,y_2,...,y_v) \leq L$ and $y_1^2 \leq M$ in $(2t)^{\Ocal(v)}\poly(H)$ time, where $H$ upper bounds the bit complexity of $L$ and $M$ (see Theorem 2.2 in \cite{RSW2016} for a restatement of this result). Under the assumptions of our lemma, $H$ is bounded by a constant.
    
    We follow the approach of \cite{RSW2016} and use binary search to determine an approximately optimal solution for our polynomial minimization problem. First, since the bit complexity of the entries of $\Ab$ are assumed to be bounded by a constant, by Corollary 38 of \cite{boutsidis2016optimal}, the objective error of the problem is either zero or greater than $2^{-\Ocal(k)}$. Therefore, we can use binary search to find a value of $L$ satisfying $\|\Xb^*\Db^* - \Ab\|_F^2 \leq L \leq (1+\epsilon)\|\Xb^*\Db^* - \Ab\|_F^2$ by running the decision algorithm of Renegar $\log 2^{\Ocal(k)} = \Ocal(k)$ times.

    Then, we can repeatedly use binary search on each variable $y_i$ with the constraints $y_i^2 \leq M$ and $P(y_1,y_2,...,y_v) \leq L$. After determining a variable $y_i$ through binary search, we can fix that variable, and then perform the procedure on the next variable. Overall, if the magnitude of the entries of $\Wb$, $\Xb^*$, and $\Db^*$, are bounded by a doubly-exponential factor of $\Ocal(nr + kd)$, we invoke the decision algorithm $2^{\Ocal(nr + kd)}$ additional times to get an overall time complexity of $2^{\Ocal(nr + kd)}$.

\end{proof}

\subsubsection{Algorithm for sparse dictionary learning}

Here, we present our algorithm for $r$-sparse dictionary learning along with a proof of its correctness and time complexity.

\begin{algorithm}[H] 
\caption{PTAS for $r$-sparse dictionary learning}
\label{alg:ptas_rsparse_dict}
\begin{algorithmic}[1]
\Require{$\Ab \in \R^{n\times d}$, $\epsilon \in (0,1)$, and $k,r \in \mathbb{N}$ such that $r \leq k$.}
\State Compute $\Ab' \in \R^{w \times d}$ and $\Wb \in \R^{w \times w}$ by the algorithm of Theorem \ref{thm:dictionary_coreset}. \label{step:dict_coreset}
\State Initialize $\Dbtil = \zero$ and $\delta = \|\Ab\|_F$. 
\For{$\Ncal \in \{(\Ncal_i)_{i \in [w]} ~|~|\Ncal_i|=r,~\Ncal_i\subset[k]\}$}
    \State Compute $\Xb', \Db' = \texttt{PolySolver}(\Ncal, \Wb, \Ab', k, r, \epsilon)$ \label{step:poly_solve}
    \If{$\|\Xb'\Db' - \Wb\Ab'\|_F < \delta$} \label{step:verify_kmeans}
        \State Set $\Dbtil = \Db'$ and $\delta = \|\Xb'\Db' - \Wb\Ab'\|_F$
    \EndIf
\EndFor
\State \Return{$\Dbtil$ and $\Xbtil = \argmin_{\Xb \in \Xcal}\|\Xb\Dbtil - \Ab\|_F$}. 
\end{algorithmic}
\end{algorithm}

\textbf{Proof of Theorem \ref{thm:sparse_dict_ptas}:}
\begin{proof}
    \textit{Correctness:} In Step \ref{step:dict_coreset} of the algorithm, by Theorem \ref{thm:dictionary_coreset}, we compute the diagonal scaling matrix $\Wb \in \R^{w \times w}$ and $\Ab'\in\R^{w \times d}$ such that, for any fixed $\Db \in \R^{k \times d}$:
    \begin{gather*}
        \Big| \min_{\Xb\in\Xcal} \|\Wb(\Xb\Db - \Ab')\|_F^2 - \min_{\Xb\in\Xcal} \|\Xb\Db - \Ab\|_F^2 \Big| \leq \epsilon\cdot\min_{\Xb\in\Xcal} \|\Xb\Db - \Ab\|_F^2.
    \end{gather*}
    Therefore, we can restrict our attention to solving for the dictionary $\Db$ that minimizes the coreset error, $\min_{\Xb\in\Xcal} \|\Wb(\Xb\Db - \Ab')\|_F^2$.
    
    At some iteration of the loop, we will guess the sparsity pattern of $\Xb^* \in \Xcal$, which we denote $\Ncal^*$.  By the guarantee of $\texttt{PolySolver}$ (Definition \ref{def:poly_solver}), $\Xb'\in\Xcal_{\Ncal^*}$ and $\Db'\in\R^{k\times d}$ computed in Step \ref{step:poly_solve} of the algorithm satisfy:
    \begin{gather*}
        \|\Wb(\Xb'\Db' - \Ab')\|_F^2 \leq (1+\epsilon) \cdot \|\Xb^*\Db^* - \Ab\|_F^2.
    \end{gather*}
    Therefore, 
    \begin{gather*}
        \min_{\Xb \in \Xcal}\|\Xb\Db' - \Ab\|_F 
        \leq (1+\epsilon)^2 \cdot \min_{\Xb\in\Xcal,\Db\in\R^{k\times d}}\|\Xb\Db - \Ab\|_F.
    \end{gather*}
    Hence, the matrices $\Dbtil$ and $\Xbtil$ achieve $\epsilon$-relative error after adjusting by a constant factor.

    \textit{Time complexity:}

    The overall time complexity of Algorithm \ref{alg:ptas_rsparse_dict} is given by:
    \begin{gather*}
    \Ocal(\text{Coreset construction})
    + |\Ncal| \times \texttt{PolySolver}\text{ time }+~\Ocal(\text{Solve for }\Xb)
    \end{gather*}

    By Theorem \ref{thm:dictionary_coreset}, the coreset construction takes $\Ocal(n^2k^{4r}2^{k^{2r+1}})$ time and $w = \Ocal(((8k^{3r}b \log d)^{\Ocal(k^{r+1})}\log n)$. The size of $\Ncal$ is $|\Ncal| = {k \choose r}^w$, and the time for one call to $\texttt{PolySolver}$ is $\Ocal(2^{\Ocal(wr + \poly(k/\epsilon))})$ by Lemma \ref{lemma:polynomial_solve_dict}. Therefore,
    \begin{gather*}
        |\Ncal| \times \texttt{PolySolver}\text{ time }
        = \exp(w \cdot r \log k) \cdot \exp(wr)
        = \exp((8k^{3r}b \log d)^{\Ocal(k^{r+1})}\log n)
    \end{gather*}
    
Finally, solving for $\Xb$ takes $n\cdot \poly(k,r,1/\epsilon)$ time, so we can ignore this term. We conclude that, overall, Algorithm \ref{alg:ptas_rsparse_dict} runs in $\exp((8k^{3r}b \log d)^{O(k^{2r+1})}\log n)$ time. Note that this is equal to $\poly(n)$ time under the assumption that $k$, $r$, $\epsilon$, and $b$ are bounded by a constant.

\end{proof}

\subsection{PTAS for $k$-means}

In this section, we provide our algorithm for $k$-means along with a proof of its correctness and time complexity. In order to improve the time complexity dependency on $k$ and $\epsilon$, we use the idea of brute force leverage score sampling, which we introduce next.

\subsubsection{Brute force leverage score sampling}
\begin{definition}\label{def:leverage_scores}
    (Leverage Score Sampling - Definition 16 in \cite{woodruff2014sketching}) Let $\Zb \in \R^{n \times k}$ have orthonormal columns, and let $p_i = \ell_i^2/k$, where $\ell_i^2 = \|\eb_i^T\Zb\|_2^2$ is the $i$-th leverage score of $\Zb$. Note that $(p_1,...,p_n)$ is a distribution. Let $\beta > 0$ be a parameter, and suppose we have any distribution $q = (q_1,...,q_n)$ for which for all $i \in [n]$, $q_i\geq\beta p_i$.

    Let $s$ be a parameter. Construct and $n \times s$ sampling matrix $\Omegab$ and an $s \times s$ rescaling matrix $\Db$ as follows. Initially, $\Omegab = \zero$ and $\Db = \zero$. For each column $j$ of $\Omegab, \Db$, independently, and with replacement, pick a row index $i \in [n]$ with probability $q_i$, and set $\Omegab_{i,j}=1$ and $\Db_{jj} = 1/\sqrt{q_i s}$.
\end{definition}

\begin{lemma}
    There is a set of matrices $\Scal \subset \R^{s \times n}$ with exactly one non-zero entry per column such that for any $\Ab \in \R^{n \times k}$ and $\Bb\in\R^{n \times d}$, there exists $\Sb \in \Scal$, so that if:
    \begin{gather*}
        \Xbtil = \argmin_{\Xb \in \R^{k \times d}} \|\Sb(\Ab\Xb - \Bb)\|_F
        ~~\text{and}~~
        \Xb^* = \argmin_{\Xb \in \R^{k \times d}} \|\Ab\Xb - \Bb\|_F,
    \end{gather*}
    then, 
    \begin{gather*}
        \|\Ab\Xbtil - \Bb\|_F \leq (1+\epsilon)\|\Ab\Xb^* - \Bb\|_F.
    \end{gather*}
    Furthermore, $\Scal$ depends only on $n$, $k$, and $\epsilon$; and $|\Scal| = n^{\Ocal(\frac{k \log k}{\epsilon})}$.
\end{lemma}
\begin{proof}
    Let $\Zb \in \R^{n \times k}$ be a matrix with orthonormal columns.  The corresponding leverage score sampling distribution $p$ satisfies $p_i = \|\eb_i^T\Zb\|_2^2/k$. We can discretize each entry $p_i$ as follows. Let $\Ical_t = [1/2^{t-1}, 1/2^t)$.  Then discretize each $p_i$ by setting $q_i = 1/2^{t-1}$ if $p_i \in \Ical_t$ for $t\leq \log n$, in which case $p_i \leq q_i \leq 2p_i$. If $p_i \not \in \cup_{t\leq \log n} \Ical_t$, then set $q_i = \frac{2}{n}$, in which case $p_i < q_i$.

    By Theorem 17 of \cite{woodruff2014sketching}, if $\Sbtil = \Omegab\Db$ is constructed as described in Definition \ref{def:leverage_scores} from the discretized distribution $q$, then for $s = \Ocal(k\log(k)/\epsilon^2)$, with at least constant probability,
    \begin{gather}\label{eqn:leverage_subspace_embedding}
        \|\Zb^T\Sbtil^T\Sbtil\Zb - \Ib\|_2 \leq \epsilon.
    \end{gather}
    This implies that there exists a fixed matrix $\Sb$ with one non-zero entry per column achieving the above error guarantee that selects $s = \Ocal(k\log(k)/\epsilon^2)$ rows of $\Zb$ and rescales the row by $1/\sqrt{q_i s}$ when the $i$-th row is selected.  Let $\Scal$ be the space of all matrices that select $s$ rows of $\Zb$ with replacement and reweights the $i$-th row according to all possible configurations of $q$.  Then, since there are $n^{\Ocal(k \log k/\epsilon^2})$ possible ways of selecting $s$ rows with replacement, and for a fixed selection of rows, the reweighting matrix $\Db$ has $(\log n)^{\Ocal(k \log (k)/\epsilon^2)}$ possibilities, $|S| = n^{\Ocal(k\log (k)/\epsilon^2)}$.

    At this point, we have shown that for parameter $\epsilon > 0$, there is a set of matrices $\Scal$ such that there exists $\Sb \in \Scal$ satisfying eqn. (\ref{eqn:leverage_subspace_embedding}), and $|S| = n^{\Ocal(k\log (k)/\epsilon^2)}$.  By setting $\epsilon' = \sqrt{\epsilon}$ in the above result, and following the proof of Theorem 3.1 in \cite{clarkson2009numerical}, we can conclude the theorem statement.
\end{proof}

\subsubsection{Algorithm for $k$-means}

Here we present our fixed-parameter PTAS for $k$-means described in Section \ref{sec:ptas_kmeans} and then provide the proof for Theorem \ref{thm:kmeans_result}.

\begin{algorithm}[H] 
\caption{PTAS for k-means}
\label{alg:ptas_kmeans}
\begin{algorithmic}[1]
\Require{Input matrix $\Ab \in \R^{n\times d}$, error tolerance $\epsilon \in (0,1)$, and number of clusters $k \in [n]$.}
\State Compute a coreset for the $k$-means problem using Algorithm 3 of \cite{bachem2018one}, denoted by the weights $\omega_1...\omega_n$, with $w = \poly(k/\epsilon)$ non-zero weights.
\State Compute a $w \times n$ matrix $\Wb$, such that if $\omega_j$ is the $t$-th non-zero weight in the coreset, then $\Wb_{tj} = \omega_t$.
\State Initialize $\Dbtil = \zero$ and $\delta = \|\Ab\|_F$. 
\For{$\Sb \in \Scal_{w,k}$}
    \For{$\Yb \in \{\Sb\Wb\Xb ~|~\Xb \in \Xcal\}$}
        \State Set $\Db' = (\Sb\Yb)^\dagger\Sb\Wb\Ab$
        \State Compute $\Xb' = \argmin_{\Xb\in \Wb\Xcal}\|\Xb\Db' - \Wb\Ab\|_F$\footnotemark
        \If{$\|\Xb'\Db' - \Wb\Ab\|_F < \delta$}
            \State Set $\Dbtil = \Db'$ and $\delta = \|\Xb'\Db' - \Wb\Ab\|_F$
        \EndIf
    \EndFor
\EndFor
\State \Return{$\Dbtil$ and $\Xbtil = \argmin_{\Xb \in \Xcal}\|\Xb\Dbtil - \Ab\|_F$}. 
\end{algorithmic}
\end{algorithm}

\footnotetext{Let $\Wb\Xcal$ denote the set $\{\Wb\Xb ~|~ \Xb\in\Xcal\}$, for the computed matrix $\Wb$.}

\textbf{Proof of Theorem \ref{thm:kmeans_result}:}
\begin{proof}
    \textit{Correctness:}
    
    In the first two steps of Algorithm \ref{alg:ptas_kmeans}, we use Algorithm 3 of \cite{bachem2018one} to compute an $\epsilon$-relative error coreset for $k$-means error. By Theorem 2 in \cite{bachem2018one}, for some $w = \poly(k/\epsilon)$, Algorithm 3 of \cite{bachem2018one} generates an epsilon relative error coreset with high constant probability. In matrix notation, this implies that their algorithm can be used to compute a matrix $\Wb \in \R^{w \times n}$ with one non-zero entry per row such that, for all $\Db \in \R^{k \times d}$,
    \begin{gather*}
        \Big|\min_{\Xb\in\Xcal} \|\Wb(\Xb\Db - \Ab)\|_F - \min_{\Xb\in\Xcal} \|\Xb\Db - \Ab\|_F\Big| \leq \epsilon \cdot \min_{\Xb\in\Xcal} \|\Xb\Db - \Ab\|_F.
    \end{gather*}
    Therefore, if $\Db'\in\R^{k \times d}$ achieves less than $(1+\epsilon)$ error on the coreset problem, then it will attain $(1+\epsilon)^2 \leq 1+3\epsilon$ error on the original problem as well. By Lemma \ref{lemma:leverage_regression}, when $\Yb = \Sb\Wb\Xb^*$,
    \begin{gather*}
        \Db' = (\Sb\Yb)^\dagger\Sb\Wb\Ab = \argmin_{\Xb \in \R^{k \times d}} \|\Sb(\Wb\Xb^*\Db - \Wb\Ab)\|_F,
    \end{gather*}
    which implies that, 
    \begin{gather*}
        \|\Wb\Xb^*\Db' - \Wb\Ab\|_F \leq (1+\epsilon)\cdot\min_{\Db\in\R^{k \times d}}\|\Wb\Xb^*\Db - \Wb\Ab\|_F.
    \end{gather*}
    Hence, in some iteration, $\Db'$ will achieve at most a $1+\epsilon$ factor error over the coreset problem, giving a $\epsilon$-relative error on the original problem after adjusting by a constant factor.

    \textit{Time complexity:}    
    
    First, by Lemma 2 of \cite{bachem2018one}, computing $\Wb$ takes $\Ocal(nkd)$ time.

    Next, by Lemma \ref{lemma:leverage_regression}, $|\Scal_{w,k}| = w^{\Ocal(\frac{k \log k}{\epsilon})} = 2^{\Ocal(\frac{k}{\epsilon}\polylog(k/\epsilon))}$. For a fixed $\Sb \in \Scal$, $|\{\Sb\Wb\Xb ~|~\Xb\in\Xcal\}| = k^{\Ocal(\frac{k}{\epsilon}\log k)} = 2^{\Ocal(\frac{k}{\epsilon}\polylog(k))}$, since there are $\Ocal(\frac{k}{\epsilon})$ rows of $\Xb$ selected by $\Sb\Wb$, and the non-zero entry in each of those rows can be in one of $k$ positions.  This implies that the inner loop of Algorithm \ref{alg:ptas_kmeans} is executed $\exp(\frac{k}{\epsilon} \polylog(k/\epsilon))$ times.

    Hence, the overall running time is $n\cdot\poly(k/\epsilon) + \exp(\frac{k}{\epsilon} \polylog(k/\epsilon))$ under our assumption that $d = \poly(k/\epsilon)$.
\end{proof}